\documentclass[sigconf,anonymous=false]{acmart}

\usepackage{enumitem}
\usepackage{multirow}
\usepackage{xspace}
\usepackage{caption}
\usepackage{subcaption}
\usepackage{algorithm}
\usepackage{algorithm}
\usepackage{algorithmicx}
\usepackage{algpseudocode}
\usepackage{amsmath}
\usepackage{verbatim}
\usepackage{silence}

\newcommand\vldbdoi{XX.XX/XXX.XX}

\newcommand\vldbavailabilityurl{URL_TO_YOUR_ARTIFACTS}
 
\newcommand{\para}[1]{\vspace{1mm}\noindent\textbf{#1}}

\newcommand\method{\textsc{FeatNavigator}\xspace}

\begin{document}
\title{\method: Automatic Feature Augmentation on Tabular Data}

\author{Jiaming Liang$^{1*}$, Chuan Lei$^2$, Xiao Qin$^2$, Jiani Zhang$^2$,\\ Asterios Katsifodimos$^{2,3}$, Christos Faloutsos$^{2,4}$, Huzefa Rangwala$^2$}
\vspace{3pt}
\affiliation{
\institution{$^1$University of Pennsylvania, $^2$AWS AI Research \& Education, $^3$TU Delft, $^4$Carnegie Mellon University\\
liangjm@seas.upenn.edu, \{chuanlei, drxqin, zhajiani, akatsifo, faloutso, rhuzefa\}@amazon.com}
}

\begin{abstract}
Data-centric AI focuses on understanding and utilizing high-quality, relevant data in training machine learning (ML) models, thereby increasing the likelihood of producing accurate and useful results. Automatic feature augmentation, aiming to augment the initial base table with useful features from other tables, is critical in data preparation as it improves model performance, robustness, and generalizability. While recent works have investigated automatic feature augmentation, most of them have limited capabilities in utilizing all useful features as many of them are in candidate tables not directly joinable with the base table. Worse yet, with numerous join paths leading to these distant features, existing solutions fail to fully exploit them within a reasonable compute budget.

We present \method, an effective and efficient framework that explores and integrates high-quality features in relational tables for ML models. \method evaluates a feature from two aspects: (1) the intrinsic value of a feature towards an ML task (i.e., feature importance) and (2) the efficacy of a join path connecting the feature to the base table (i.e., integration quality). \method strategically selects a small set of available features and their corresponding join paths to train a feature importance estimation model and an integration quality prediction model. Furthermore, \method's search algorithm exploits both estimated feature importance and integration quality to identify the optimized feature augmentation plan. Our experimental results show that \method outperforms state-of-the-art solutions on five public datasets by up to 40.1\% in ML model performance.
\end{abstract}

\maketitle



\section{Introduction}
\label{sec:intro}

Feature augmentation is a crucial process that enhances machine learning (ML) model performance~\cite{dfs,arda,autof,metam}. The process of feature augmentation on tabular data typically involves feature exploration and integration. Feature exploration identifies useful features among candidate tables that could reveal additional insights or patterns. And feature integration decides how to join the selected candidate tables with the base table to augment the initial training data. Open tabular data repositories such as Google Dataset Search\footnote{\url{https://datasetsearch.research.google.com/}}, Kaggle\footnote{\url{https://www.kaggle.com/datasets}}, and OpenML\footnote{\url{https://www.openml.org/}} provide diverse tabular data that capture a wide range of features. By exploring and integrating features from these repositories of data, ML models can generalize to more patterns~\cite{DBLP:journals/jbd/ChenDHC20}, overcome overfitting~\cite{ying2019overview}, thereby leading to improved performance on unseen data~\cite{arda,autof,metam}.

\begin{figure}
  \centering
  \includegraphics[width=\columnwidth]{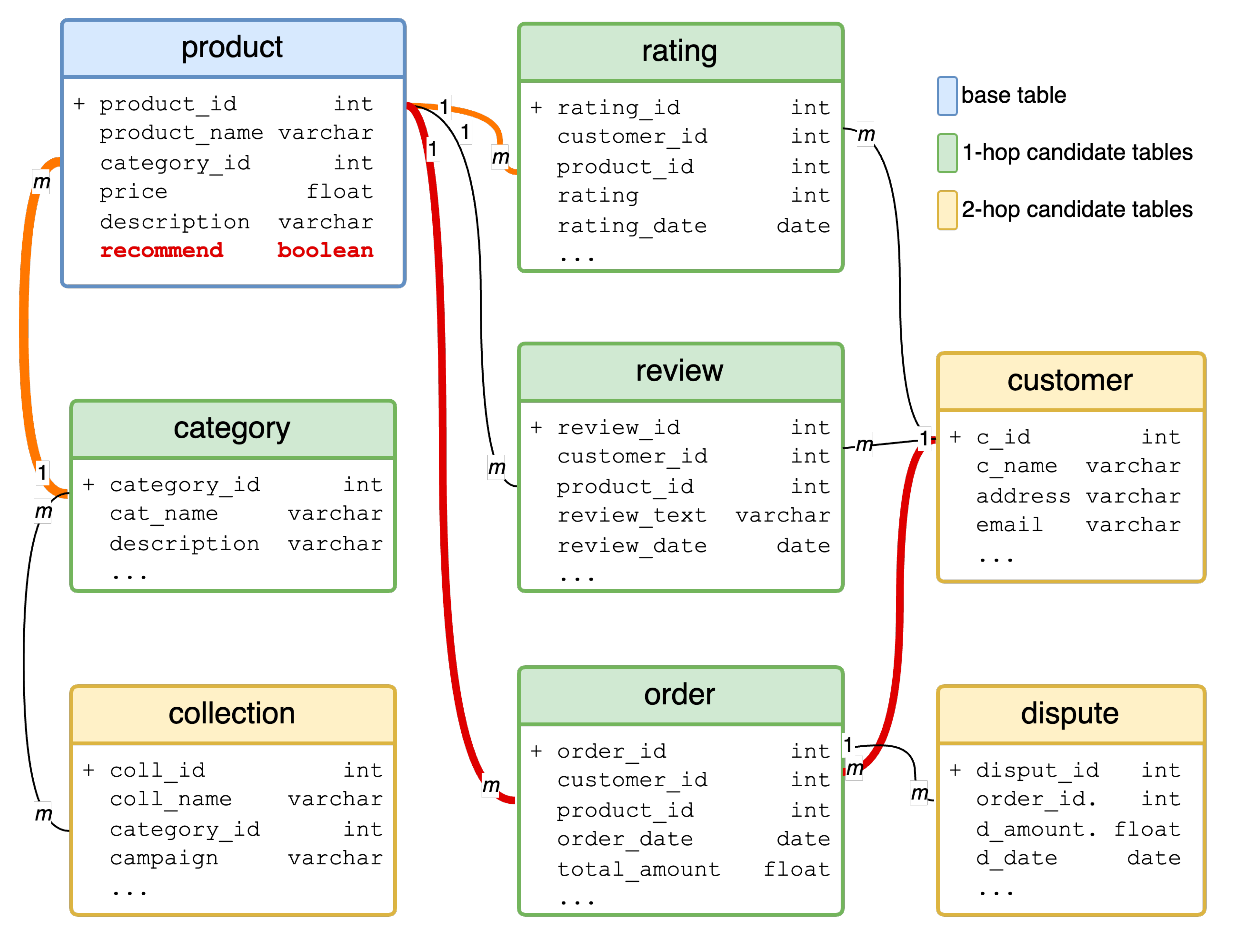}
  \caption{Example of feature augmentation. The schema of the base table, \textsf{product}, is in \textcolor{blue}{blue}. The task is to predict the values in the boolean type column \textcolor{red}{`recommend'}. The schemata in \textcolor{teal}{green} and \textcolor{yellow}{yellow} describe the tables that are 1-hop and 2-hop joinable with the base table.}
  \label{fig:motivation}
  \vspace{-10pt}
\end{figure}

\para{Motivating Example.} Figure~\ref{fig:motivation} shows an example of feature augmentation on a toy dataset with a classification task. The base table is \textsf{product}, in which \textsf{recommend} is the target column, indicating whether the product should be recommended or not. The tables \textsf{category}, \textsf{rating}, \textsf{review}, and \textsf{order} are \textit{directly} joinable with \textsf{product}. And \textsf{collection}, \textsf{dispute} and \textsf{customer} are candidate tables \textit{indirectly} joinable with \textsf{product} through the other tables. Assuming that the features (e.g., \textsf{address}) from \textsf{customer} are useful to the prediction task, three join paths can be used to augment the base table \textsf{product}. However, these paths are not equivalent as customers may not provide ratings or reviews for all purchased products. In fact, the join path, \textsf{customer} $\bowtie$ \textsf{order} $\bowtie$ \textsf{product} (indicated by the thick red line), provides the highest quality to augment the base table. One can only discover that by actually joining these tables, which can be expensive when involving multiple large tables. When the prior knowledge about feature importance is missing, it becomes more challenging to examine all candidate tables, including both directly and indirectly joinable ones. Expensive joins need to be executed in order to find out the actual benefit of an augmented table to the prediction task. Worse yet, certain joins (e.g., \textsf{product} $\bowtie$ \textsf{category} and \textsf{product} $\bowtie$ \textsf{rating} indicated by the thick orange lines) are repeated multiple times when exploring different tables.

\para{State-of-the-art Approaches.} 
Early works~\cite{10.1145/2882903.2882952,arda} augment the base table by joining other tables based on either their relevancy to the base table or the improvement of the contained features bring to the ML models. However, they assume that the features and their respective relevance to the base table are specified. Hence they fail to leverage all possible candidate feature and consequently do not reach the optimal model performance. Recently, ALITE~\cite{DBLP:journals/pvldb/KhatiwadaSGM22} and DIALITE~\cite{DBLP:conf/sigmod/KhatiwadaSM23} are introduced to integrate tables based on joinability or unionability, which can improve the quality of downstream applications (e.g., ML models).

Reinforcement learning (RL)~\cite{DBLP:conf/ecml/VermorelM05,10.5555/3045390.3045601} has emerged as an effective approach to balance between feature exploration and exploitation. Several studies~\cite{autof,10.14778/3523210.3523223,metam} leverage RL techniques, such as multi-armed bandits and Deep Q Networks (DQNs), to explore the large search space of candidate tables. Although these RL-based methods have achieved promising results, they still face critical challenges. Specifically, METAM~\cite{metam} is limited to direct joinable candidate tables and AutoFeature~\cite{autof} does not provide fine-grained optimization on integration quality. Moreover, RL often relies on a large number of samples to learn effectively (DQNs in particular), and tends to overfit to the training environment, leading to poor performance in slightly different real-world situations.

\para{Challenges.} Based on the above observations, we aim to address the following open challenges.
\begin{enumerate}[leftmargin=*]
    \item Integration Quality. In real-world open data settings, a candidate table can be integrated with the base table via multiple join paths, which often have very different integration quality. Hence, it is critical to identify the optimal join path to achieve the full potential (i.e., feature importance) of a candidate table.
    \item Expensive Join Operations. Join operations are computationally expensive and time consuming, especially when a long join path involves multiple joins. Candidate tables from open data repositories often share similar join paths. Therefore, avoiding or sharing join operations is imperative to efficient feature augmentation.
    \item Search Complexity. The search space of all candidate tables along with all integration paths is intractable, making it prohibitively expensive to find the optimal set of candidate tables. To address this challenge, we need effective strategies to reduce the search space without compromising optimality.
\end{enumerate}

\para{Proposed Solution.} To cope with the above mentioned challenges, we propose \method, an efficient and effective framework for automatic feature augmentation. Given a base table, a set of candidate tables and a ML task, \method first effectively selects a subset of candidate tables using a lightweight clustering algorithm based on table characteristics (e.g., embeddings, statistics, semantics, etc.). \method further chooses representative join paths from these candidate tables to the base table. This allows \method to not only find out the actual performance gain led by the candidate tables following the join paths but also learn an integration quality model and a feature importance estimator to calculate the performance gain of any unexplored candidate tables and their respective join paths. Lastly, \method's path search algorithm utilizes the integration quality model and feature importance estimator to further prune the search space and finds the optimized integration paths for the selected candidate tables in polynomial time. In case of having new tables added or table updates to the data repository, \method's integration quality model and feature importance estimator can effectively and efficiently estimate the performance gain without retraining. In polynomial time, \method's search algorithm produces new optimized integration paths if the estimated performance gain is higher than the one of the existing paths.


\para{Contributions.} Our main contributions are the following.

\begin{enumerate}[leftmargin=*]
    \item We introduce \method, a novel feature augmentation framework designed for ML tasks over tabular data repositories.
    \item We decompose the intractable problem of feature augmentation into two sub-problems, feature importance and integration quality, and consequently prove the relationship between them.
    \item We design an effective feature exploration module consisting of a lightweight clustering-based feature importance estimation and an effective LSTM-based integration quality model.
    \item We propose an efficient feature path search algorithm that exploits both estimated feature importance and integration quality to identify high-quality features and their optimized join paths for integration.
    \item Our experiments show that \method outperforms state-of-the-art baselines on different ML tasks over five public datasets with up to 40.1\% improvement ML model performance.
\end{enumerate}

\section{Preliminaries and System Overview}
\label{sec:prel}

In this section, we first discuss the notions used in this paper, formally define the problem of feature augmentation and then provide an overview of \method.

\subsection{Basic Notions}
\label{sec:prel:notions}

\begin{definition}[Base Table]
The base table, $T_\textit{base}$ = \{$c_1$, $c_2$, \ldots, $c_n$, $c_l$\}, refers to the original table that includes the initial set of features and the label column before any augmentation. Each $c_i$ ($i\in$ [1, $n$]) represents a column and $c_l$ denotes the label column, which stores the prediction target for each row.
\end{definition}

\begin{definition}[Candidate Tables]
The candidate tables, $\mathcal{T}$ = \{$T_1$, $T_2$, \ldots, $T_m$\}, are $m$ joinable tables that can be used to augment a base table $T_{base}$. Each table is denoted by $T_j$ = \{$f_{j,1}$, $f_{j,2}$, \ldots, $f_{j,k}$\} ($j\in$ [1, $m$]), where $f_{j,k}$ denotes the $k$-th column (a.k.a feature)\footnote{We use two terms, column and feature, alternately in this paper.} in table $T_j$.
\end{definition}

\begin{definition}[Join Path]
Given a set of tables \( T_1, T_2, \ldots, T_n \), a join path \( P \) can be defined as a sequence of join operations between pairs of tables. A join operation between any two tables \( T_i \) and \( T_j \) can be denoted as \( T_i \bowtie T_j \). The join path \( P \) can be represented as: \( P = T_1 \bowtie T_2 \bowtie T_3 \ldots \bowtie T_n \).
\end{definition}

\begin{figure*}
  \centering
  \includegraphics[width=0.95\linewidth]{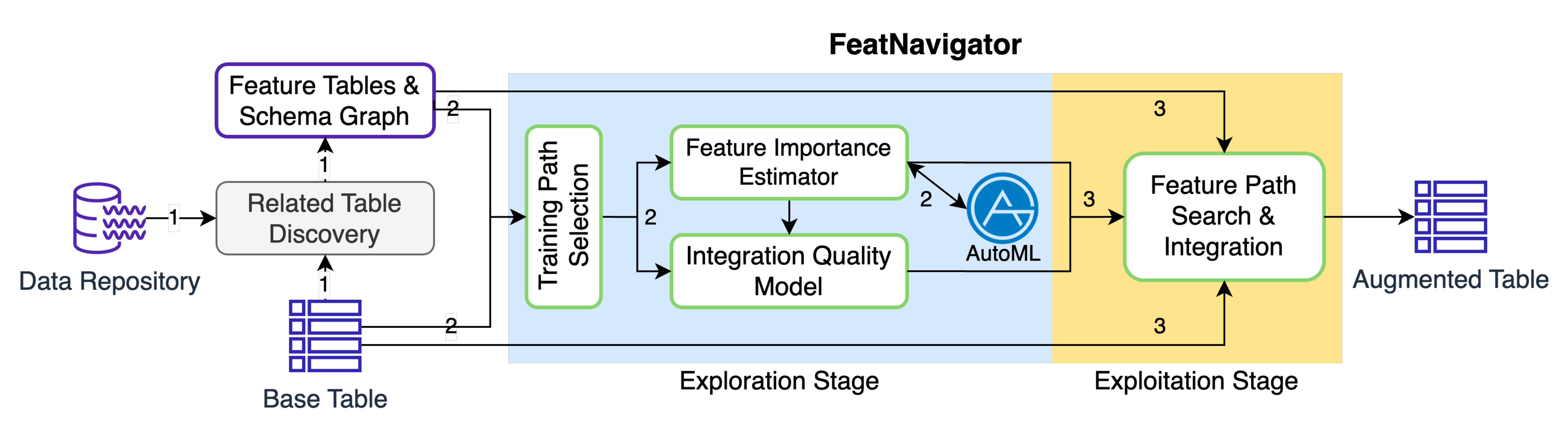}
  \caption{\method\ overview.}
  \label{fig:sys}
\end{figure*}

Hence a join path between the base table $T_{base}$ and a candidate table $T_c$ can be expressed as: \( P = T_{base} \bowtie T_i \bowtie T_{i+1} \ldots \bowtie T_c \). Without loss of generality, we allow different types of join operators between the base table \( T_{base} \) and a candidate table \( T_c\in \mathcal{T} \), including inner join, left outer join, right outer join, etc.

\begin{definition}[Augmented Table]
The augmented table \( T_{aug} \) is the result of joining the base table \( T_{base} \) and a candidate table \( T_c \) following a join path \( P \).
\label{def:taug}
\end{definition}

Note that there can be multiple join paths between the base table and a candidate table. Hence the augmented table could be different depending on the chosen join path.

\begin{definition}[Machine Learning Task]
A machine learning (ML) task is defined as a function \( f \) that maps an input space \( \mathcal{X} \) to an output space \( \mathcal{Y} \). Given a dataset \( D = \{ (x_1,y_1), (x_2,y_2), \ldots, (x_n,y_n) \} \), where \( x_i \in \mathcal{X} \) and \( y_i \in \mathcal{Y} \), the objective is to learn the function \( f \) that best approximates the true underlying relationship from the inputs to the outputs. For a classification task, we have \( f: \mathcal{X} \to \mathcal{Y} \), where \( \mathcal{Y} \) is a discrete set of labels. For a regression task, we have \( f: \mathcal{X} \to \mathbb{R} \), where \( \mathcal{Y} \) is a continuous domain.
\end{definition}

\begin{definition}[Utility Gain]
Given a ML task, the utility score, \( US(T) \), is defined as the objective or performance metric (e.g., accuracy, precision, recall, F1 score, etc.) of the task on a table $T$ achieved by a specific ML model. Note that \( US(T) \) uses the negation of a metric value when smaller values are preferred (e.g., mean squared error (MSE), mean absolute error (MAE), etc.). Let \( US(T_{base}) \) be the performance metric on the base table, and \( US(T_{aug}) \) be the performance metric on the augmented table. The utility gain \( UG \) is defined as the improvement in the performance metric: \( UG = US(T_{aug}) - US(T_{base}) \).
\label{def:ug}
\end{definition}

\para{Problem Statement.}
Given a base table \( T_{base} \), a set of candidate tables \( \mathcal{T} \), and a ML task, our goal is to select a subset of tables from \( \mathcal{T} \) and to identify the best join paths to augment with \( T_{base} \), such that the utility gain \( UG \) of the ML task on \( T_{aug} \) is maximized.

The complexity of choosing a subset of tables (i.e., feature selection) generally increases exponentially with the number of candidate tables (features), as all possible combinations of features need to be evaluated, making the problem NP-hard. Even with various heuristic and optimization techniques~\cite{DBLP:conf/icml/JohnKP94,298224} to reduce the complexity, the cost of feature selection can be still expensive. Worse yet, enumerating all possible join paths among the base and candidate tables is prohibitively expensive due to its exponential time complexity in the worst case. The combination of these two exponential search spaces makes the problem intractable, consequently needing an effective and efficient method.


\subsection{\method\ System Overview}
\label{sec:prel:overview}

For a given ML task, \method takes input as a base table with a prediction task on the target column and a tabular data repository consisting of candidate tables with features. \method leverages data discovery systems such as Aurum~\cite{aurum} or NYU Auctus~\cite{auctus} to identify the candidate tables within $k$-hop joins to the base tables. \method\ then triggers its utility gain model in two stages.

In the exploration stage, \method takes the following four steps. (1) It first groups the identified features into clusters based on their representations (embeddings, data distributions, types, etc.) using an efficient clustering algorithm. (2) To train the integration quality model, \method employs a path selection strategy to choose a number of features and candidate join paths. 
(3) \method then follows the selected join paths to integrate the target features into the base table and runs AutoML (e.g., AutoGluon\footnote{\url{https://auto.gluon.ai/}}) to find out the actual utility gain compared to the model performance only using the base table. (4) \method effectively collects a sufficient amount of training data to train the integration quality model and also assigns the actual feature importance to the selected features in the clusters. The importance of the other features can be estimated based on their relative distance to the ones with the actual importance assigned.

In the exploitation stage, \method leverages both a feature importance estimator and an integration quality model to conduct an efficient and effective search over all identified features and their associated join paths. It iteratively selects the most important feature and uses the integration quality model to select the optimal join path (i.e., best integration quality) until a predefined number of features is reached. 
\method then follows the selected join paths to integrate the selected features with the base table and generates the final augmented table for the given ML task. Finally, AutoML solutions (e.g., AutoGluon) can train, select and deploy the best performing machine learning model on the augmented table~\cite{erickson2020autogluon}.

\section{Utility Gain Modeling}
\label{sec:decomposition}

As discussed in Section~\ref{sec:prel:notions}, the problem of \textit{maximizing utility gain} for a given ML task is intractable. Hence we propose to decompose this problem into two sub-problems, namely maximizing feature importance and maximizing integration quality. The motivation behind this decomposition is that the effectiveness of a candidate feature relies on two critical factors, namely (1) feature importance -- how useful is a feature to the ML model if it can be fully leveraged (e.g., the significance of product ratings in the context of product recommendations), and (2) integration quality -- how many data instances in the base table can be augmented by a feature from a candidate table. In the rest of this section, we will introduce several key concepts to help define these two factors and use them together to derive the decomposition and our final optimization goal. The detailed algorithms for achieving the derived optimization goal are then introduced in the next section.


\para{Feature Importance.} To define the \textit{feature importance}, we first introduce the concept of a \textit{virtual table}. Let \( T_{aug} \) denote a resulting augmented table from a base table \( T_{base} \) and a candidate table \( T_{c} \). The virtual table \( T_{virt} \) refers to the ideal case where there exists a join path $P$ such that every instance in \( T_{base} \) can be augmented by a feature instance from \( T_{c} \) via this $P$ (i.e., iPhone 15 with a rating of 5, Pixel 8 with a rating of 4.5, and Galaxy S23 with a rating of 4 shown in Figure~\ref{fig3:virtual}). Note that a virtual table is specific to a base table and an individual candidate feature. 

\begin{definition}[Feature Importance]
\label{def.fi}
Feature importance is defined as the largest utility gain achievable by augmenting the base table with a candidate feature $f_{i,j} \in T_i$, i.e., 
\[ FI(f_{i,j}) = US(T_{virt(i,j)}) - US(T_{base}), \]
where $f_{i,j}$ is the $j$-th feature from the table $T_i$, $T_{virt(i,j)}$ is the virtual table augmented by the feature $f_{i,j}$ and $FI(f_{i,j})$ ranges from -1 to 1.
\end{definition}

$FI(\cdot)$ quantifies the (positive or negative) impact of a candidate feature to the ML task on the base table. Note that the value of $FI(\cdot)$ depends on the corresponding metric used by the utility score.

\begin{figure}
\begin{subfigure}{\columnwidth}
  \centering
  \includegraphics[width=\linewidth]{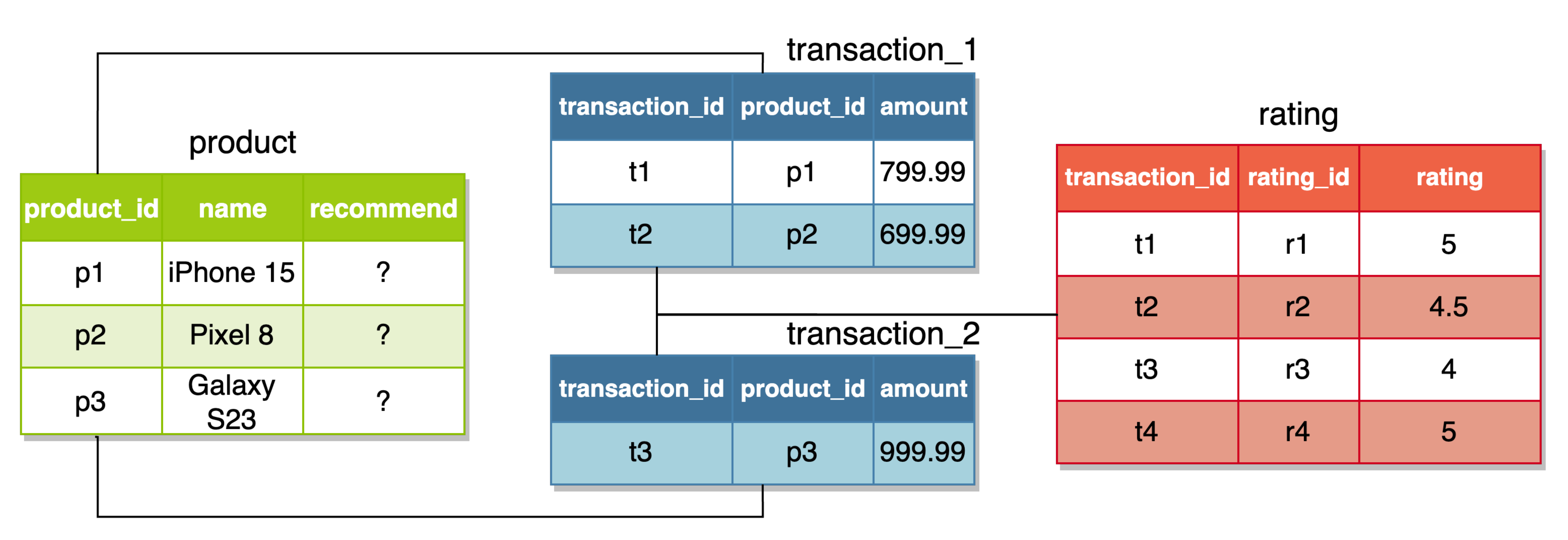}
  \caption{Join graph}
  \label{fig3:sg}
\end{subfigure}
\newline
\begin{subfigure}{.32\columnwidth}
  \centering
  \includegraphics[width=\linewidth]{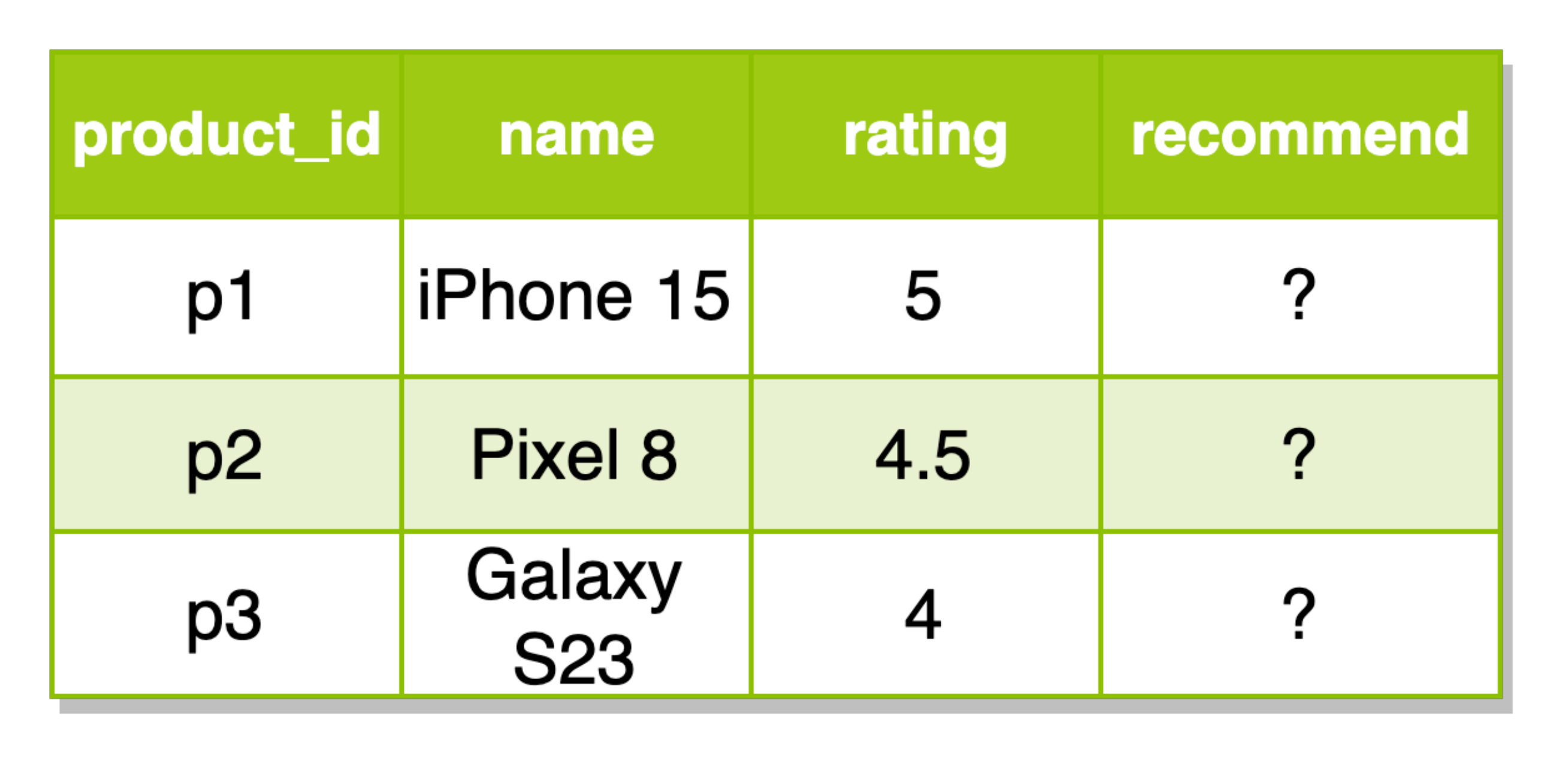}  
  \caption{$T_{\textit{virt}}$}
  \label{fig3:virtual}
\end{subfigure}
\begin{subfigure}{0.33\columnwidth}
  \centering
  \includegraphics[width=\linewidth]{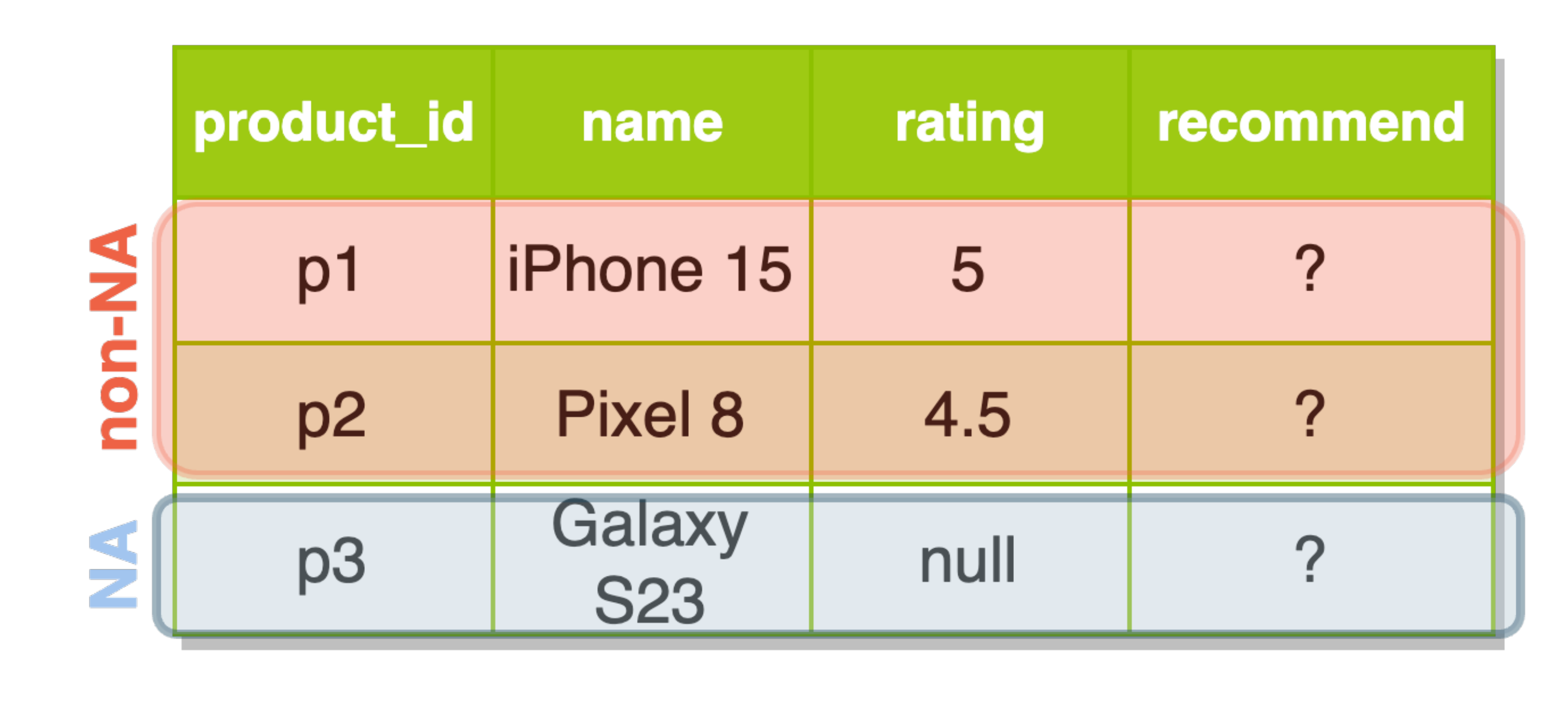}  
  \caption{$T_{\textit{aug\_1}}$}
  \label{fig3:aug1}
\end{subfigure}
\begin{subfigure}{.33\columnwidth}
  \centering
  \includegraphics[width=\linewidth]{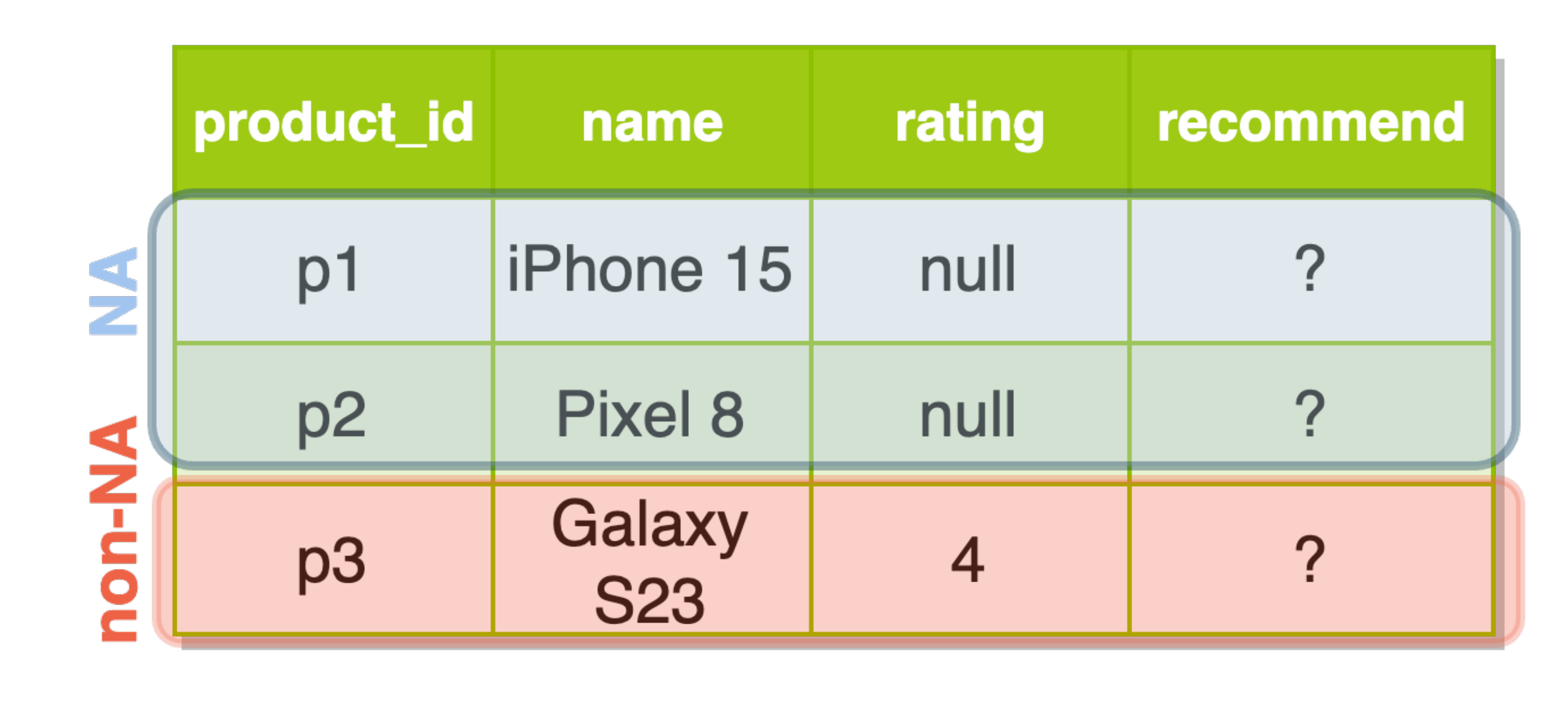}  
  \caption{$T_{\textit{aug\_2}}$}
  \label{fig3:aug2}
\end{subfigure}
\vspace{-15pt}
\caption{A join graph with examples of augmented tables.}
\label{fig3}
\end{figure}

\para{Integration Quality.} On the other hand, integration quality measures the number of feature instances that can be carried over into the base table via a particular join path. It quantifies the proportion of instances in the base table that can be successfully augmented by the candidate table. For example, Figures~\ref{fig3:aug1} and~\ref{fig3:aug2} show two augmented tables using the same feature \textsf{rating}. The join path $P_1$ = \textsf{product} $\bowtie$ \textsf{transaction\_1} $\bowtie$ \textsf{rating} brings 2 out of 4 instances from the \textsf{rating} table to the \textsf{product} table. However, only 1 instance from the \textsf{rating} table can be augmented to the base table by following $P_2$ = \textsf{product} $\bowtie$ \textsf{transaction\_2} $\bowtie$ \textsf{rating}. Hence $P_1$ is a better integration path than $P_2$. 

When data instances in the base table cannot be augmented by a candidate feature, we use the special token \textsf{null} to fill the missing feature value, which will be used for the modeling training and evaluation. Note that AutoML solutions such as AutoGluon can automatically handle missing values during model training~\cite{erickson2020autogluon}.
Therefore, an augmented table can be split into two parts, namely \textit{non-NA} and \textit{NA}. \textit{non-NA} refers to the set of instances in $T_{aug}$ that are augmented by feature instances, and \textit{NA} denote the data instances in $T_{aug}$ filled with \textsf{null} values. For example, in Figure~\ref{fig3:aug1}, two products $p1$ and $p2$ in the red box belong to \textit{non-NA}, and the other product $p3$ composes the \textit{NA}.

\begin{definition}[Integration Quality]
Integration quality of an augmented table \( T_{aug} \) is defined as a ratio between the number of feature instances that contribute to the augmentation and the total number of instances in the base table through a specific join path: 
\[ IQ(T_{aug}) = \frac{|\textit{non-NA}|}{|T_{aug}|} = \frac{\#\ augmented\ instances}{\#\ instances} \text{ in } T_{aug}, \]
where \( T_{aug} = Augment(T_{base}, f_{i,j}, P_k) \) indicates that $T_{aug}$ is produced by augmenting the feature $f_{i,j}$ to the base table, following the path $P_k$. Intuitively, $IQ(T_{aug})$ ranges from 0 to 1.
\label{def:iq}
\end{definition}

\para{Decomposition of \textit{UG}.} We argue that the \textit{utility gain} ($UG$) can be expressed by the multiplication of \textit{feature importance} ($FI$) and \textit{integration quality} ($IQ$), i.e., \( UG(T_{aug}) = FI(f_{i,j}) \times IQ(T_{aug}) \). To prove this hypothesis, we will discuss (1) how to estimate the utility score of an augmented table with missing feature values (Theorem~\ref{lm.weightedUG}), and (2) the error bound of utility score estimation (Theorem~\ref{lm.error_bound}).



\begin{lemma}
\label{lm.weightedUG}
\( US(T_{aug}) \) is the weighted average of
\( \ US(\textit{non-NA}) \) and \( US(\textit{NA}), \) 
namely, \( US(T_{aug}) = p \times US(\textit{non-NA}) + (1-p) \times US(\textit{NA}) \), where $p$ denotes the integration quality measurement for simplicity purposes.
\end{lemma}


\begin{proof}
Without loss of generality, we choose the classification accuracy as the utility score. By definition, the classification accuracy for \textit{non-NA} and \textit{NA} in $T_{aug}$ is:
\[ US(\textit{non-NA}) =\frac{CP_{\textit{non-NA}}}{|\textit{non-NA}|}\text{ and } US(\textit{NA})=\frac{CP_{\textit{NA}}}{|NA|}, \]
where $CP$ indicates the correct predictions including true positives and negatives, $|\textit{non-NA}|$ and $|\textit{NA}|$ are the cardinality of \textit{non-NA} and \textit{NA} in $T_{aug}$. Consequently, the accuracy of the augmented table $T_{aug}$ can be expressed as follows:

\begin{equation}
\begin{aligned}
US(T_{aug}) & = \frac{CP_{T_{aug}}}{|T_{aug}|} = \frac{CP_{\textit{non-NA}} + CP_{\textit{NA}}}{|T_{aug}|}, \\
            & = \frac{CP_{\textit{non-NA}}}{|\textit{non-NA}|} \times \frac{|\textit{non-NA}|}{|T_{aug}|} + \frac{CP_{\textit{NA}}}{|\textit{NA}|} \times \frac{|\textit{NA}|}{|T_{aug}|},\\
            & = p \times US(\textit{non-NA}) + (1-p) \times US(\textit{NA}).
\end{aligned}
\end{equation}
\end{proof}

\begin{lemma}
\label{lm.expectation_variacnce}
If a random variable x, s.t. \( -1 < x < 1, E(x)=0, Var(x) \leq \epsilon < 1 \), then \( 0 \leq E(|x|) \leq 2\sqrt{\epsilon} \) and \( 0 \leq Var(|x|) \leq \epsilon \).
\end{lemma}

\begin{proof}
To begin with, we have 
\begin{equation}
E(x^2) = Var(x) - E(x)^2 \leq \epsilon
\label{eq.lm.expectation_variance.1}
\end{equation}
Let $a$ be a constant \( 0 < a < 1 \), and we have \( E(|x|)=E(\sqrt{x^2}) \). For \( a<|x|<1 \), we have \( \sqrt{x^2}<\frac{1}{a}x^2 \), and for \( 0 \leq |x| \leq a \), we have \( \sqrt{x^2} \leq a \). Consequently, \( \sqrt{x^2} \leq a+\frac{1}{a}x^2 \). Thus for $E(|x|)$, we have
\begin{equation}
\label{eq.lm.expectation_variance.2}
\begin{aligned}
E(|x|)  & \leq E(a+\frac{1}{a}x^2) \\
        & =a+\frac{1}{a}E(x^2) \\
\end{aligned}
\end{equation}

Based on Equations~\ref{eq.lm.expectation_variance.1} and~\ref{eq.lm.expectation_variance.2}, we have 
\( E(|x|) \leq a+\frac{1}{a}\epsilon \), and
let \( a=\sqrt{\epsilon} \), we have \( E(|x|) \leq 2\sqrt{\epsilon}\). For $Var(|x|)$, we have
\begin{equation}
\begin{aligned}
Var(|x|) & = E(|x|^2)-E(|x|)^2 \leq E(|x|^2) \\
         & = E(x^2) \leq \epsilon \\
\end{aligned}
\end{equation}

Given that \( 0 \leq E(|x|) \) and \( 0 \leq Var(|x|) \), we prove that \( 0 \leq E(|x|) \leq 2\sqrt{\epsilon} \) and \( 0 \leq Var(|x|) \leq \epsilon \).
\end{proof}

\begin{lemma}
\label{lm.error_bound}
The utility score of \textit{non-NA} part of $T_{aug}$ can be estimated by the one of\ \ \( T_{virt} \) with an error. The expectation and variance of the absolute error are bounded by the reciprocal of the number of rows in \( T_{virt} \). Namely,
\[ error = (US(\textit{non-NA}) - US(T_{virt})) \times \frac{|\textit{non-NA}|}{|T_{virt}|}, \]
with \( 0 \leq E(|error|) \leq \frac{2}{\sqrt{|T_{virt}|}} \) and \( 0 \leq VAR(|error|) \leq \frac{1}{|T_{virt}|} \). Thus \( E(|error|) \approx 0 \) and \( Var(|error|) \approx 0 \), when \( |T_{virt}| \to +\infty \).
\end{lemma}

\begin{proof}
Without loss of generality, we again choose the classification accuracy as the utility score. Let us denote $N$ as the total number of instances (rows) in \( T_{virt} \), $M$ as the number of correct predictions (i.e., true positives and true negative) when the model is trained on \( T_{virt} \), $n$ as the number of instances (rows) in \textit{non-NA}, and $m$ as the number of correct predictions when the model is trained on \textit{non-NA}. 

We express the accuracy attained through training on \( T_{virt} \) and \textit{non-NA} as \( US(T_{virt})=\frac{M}{N} \) and \( US(\textit{non-NA})=\frac{m}{n} \), respectively. The approximation error, by definition, is expressed as:
\begin{equation}
\begin{aligned}
error & = US(\textit{non-NA}) - US(T_{virt}) \times \frac{|\textit{non-NA}|}{|T_{virt}|}, \\
      & = \frac{n}{N}(\frac{m}{n}-\frac{M}{N}).
\end{aligned}
\label{eq:error1}
\end{equation}

Let us assume that $n$ instances are randomly selected from $N$ instances, implying that $m$ follows a hypergeometric distribution, i.e., $Pr(m=x) = \frac{\binom{M}{x} \binom{N-M}{n-x}}{\binom{N}{n}}$. The expectation of $m$ is $E(m)=\frac{nM}{N}$ and the variance of $m$ is $Var(m)=\frac{nM(N-n)(N-M)}{N^2(N-1)}$. Then we have
\begin{equation}
\begin{aligned}
E(\frac{n}{N}(\frac{m}{n} - \frac{M}{N})) & = 0,\\
Var(\frac{n}{N}(\frac{m}{n} - \frac{M}{N})) & = \frac{nM(N-n)(N-M)}{N^4(N-1)}.
\end{aligned}
\label{eq:error2}
\end{equation}

According Equations~\ref{eq:error1} and~\ref{eq:error2}, we can infer that \( E(error) = 0 \) and \( Var(error)=\frac{nM(N-n)(N-M)}{N^4(N-1)} \). Furthermore, as $N>M$, we have \( \frac{M(N-n)(N-M)}{N^2(N-1)}\leq 1 \), and hence \( Var(error) \leq \frac{n}{N^2} \leq \frac{1}{N}<1 \). By setting \( \epsilon=\frac{1}{N} \), we have \( E(error)=0 \) and \( Var(error) \leq \epsilon < 1\ (-1<error<1) \). Based on Lemma~\ref{lm.expectation_variacnce}, we have \( E(|error|) \leq 2\sqrt{\epsilon} = \frac{2}{\sqrt{N}} \) and \( Var(|error|) \leq \epsilon = \frac{1}{N} \). Subsequently, we have \( 0 \leq E(|error|) \leq \frac{2}{\sqrt{N}} \) and \( 0 \leq Var(|error|) \leq \frac{1}{N} \). Hence we prove \( E(|error|) \approx 0 \) and \( Var(|error|) \approx 0 \) when \( N \to +\infty \).
\end{proof}

Based on Lemma~\ref{lm.weightedUG} and Lemma~\ref{lm.error_bound}, Theorem~\ref{th:fi&iq} below states the relationships among utility gain, feature importance, and integration quality.

\begin{theorem}
\label{th:fi&iq}
The utility gain of an augmented table \( T_{aug} = Augment(T_{base}, f_{i,j}, P_k) \) can be approximated by a multiplication of the feature importance of $f_{i,j}$ and the integration quality of $f_{i,j}$ following the join path $P_k$, namely \( UG(T_{aug}) = FI(f_{i,j}) \times IQ(T_{aug}) \).
\end{theorem}

\begin{proof}
Theorem~\ref{th:fi&iq} can be proved by using Definition~\ref{def:ug}, Lemma~\ref{lm.weightedUG}, and Lemma~\ref{lm.error_bound}.
\small
\begin{equation}
\label{eq.ugfiiq}
\begin{aligned}
UG(T_{aug}) & = US(T_{Aug}) - US(T_{base}), \\ 
            & = p \times US(\textit{non-NA}) + (1-p) \times US(\textit{NA}) - US(T_{\textit{base}}), \\ 
            & \approx p \times US(T_\textit{virt}) + (1-p) \times US(T_{base}) - US(T_{\textit{base}}), \\
            & = p \times US(T_\textit{virt}) - p \times US(T_\textit{base}), \\ 
            & = p \times (US(T_\textit{virt}) - US(T_\textit{base})), \\ 
            & = FI(f_{i,j}) \times IQ(T_{aug}). 
\end{aligned}
\end{equation}
\normalsize
\end{proof}

\begin{figure*}
\begin{subfigure}{.30\textwidth}
    \centering
    \includegraphics[width=\linewidth]{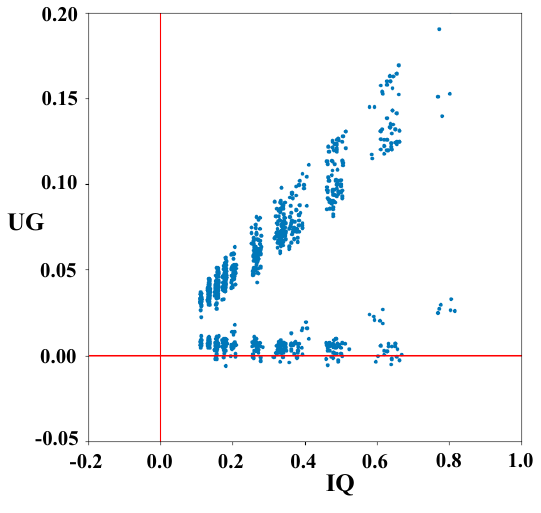}  
    \caption{Relationship between {\scriptsize $IQ$} and {\scriptsize $UG$}}
    \label{fig4:iqug}
\end{subfigure}
\begin{subfigure}{.30\textwidth}
    \centering
    \includegraphics[width=\linewidth]{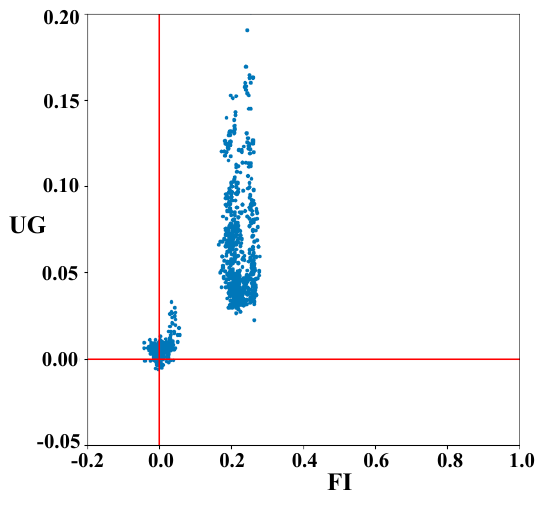}  
    \caption{Relationship between {\scriptsize $FI$} and {\scriptsize $UG$}}
    \label{fig4:fiug}
\end{subfigure}
\begin{subfigure}{.30\textwidth}
    \centering
    \includegraphics[width=\linewidth]{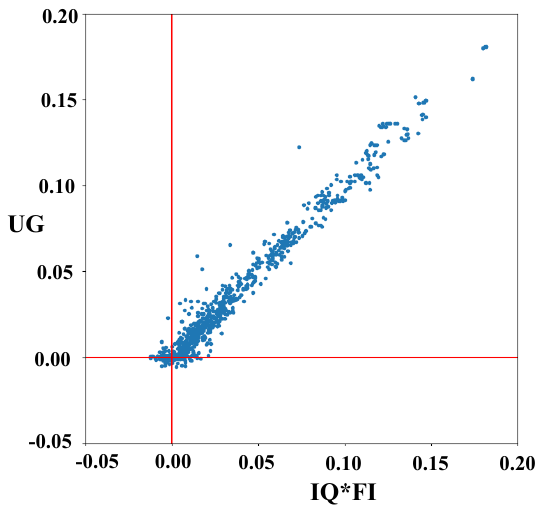}  
    \caption{Linear relationship between {\scriptsize  $IQ{\cdot}FI$ and $UG$}}
    \label{fig4:fiiqug}
\end{subfigure}
\caption{Utility gain, feature importance, and integration quality.}
\label{fig4}
\end{figure*}

\para{Empirical Validation.}
We conduct an empirical study to validate the derived relationship in Theorem~\ref{th:fi&iq}. Specifically, we choose \textsf{DonorsChoose} dataset\footnote{\url{https://www.kaggle.com/c/donorschoose-application-screening}} from Kaggle and randomly sample 1,000 candidate features and their corresponding join paths. We use these sampled features to augment the base table and compute the metrics $UG$, $IQ$, and $FI$, subsequently plotting the data points in Figure~\ref{fig4}. We have two observations: (1) lower $IQ$ and $FI$ values lead to lower $UG$ values, and (2) $IQ\cdot FI$ has a linear relationship with $UG$, which are consistent with Theorem~\ref{th:fi&iq}.

\begin{theorem}
\label{th:ugiqfi}
The problem of utility gain maximization
\[ \max UG(T_{aug}) = \max \sum_{T_i \in \mathcal{T'},\ P_k \in \mathcal{P'}} UG(T_{aug_i})\]
can be approximated by
\[ \max \sum_{T_i \in \mathcal{T}',\ P_k \in \mathcal{P'}} IQ(T_{aug_i}) \times FI(f_{i,j}), \]
where $f_{i,j}\in T_i$ denotes a candidate feature $f_{i,j}$ in $T_i$ and $T_{aug_i}$ represents the augmented table \( Augment(T_{base}, f_{i,j}, P_k) \).
\end{theorem}

With Theorem~\ref{th:ugiqfi}, we establish a theoretical foundation to transform the problem of \textit{utility gain maximization} into two sub-problems that are more tractable individually. In the following sections, we will describe integration quality and feature importance estimation and an optimized search algorithm that effectively prunes the feature augmentation space based on the properties of integration quality and feature importance.

\section{Integration Quality \& Feature Importance Estimation}
\label{sec:estimation}

The utility gain has been shown to have a linear relationship with integration quality ($IQ$) and feature importance ($FI$) in Theorem~\ref{th:fi&iq}. However, both $IQ$ and $FI$ are computational expensive to obtain as they require resource intensive joins and model training, particularly when dealing with a large number of features and their associated join paths. As a result, estimating $IQ$ and $FI$ to avoid costly joins and model training becomes imperative. In this section, we delves into our proposed solution for $IQ$ and $FI$ estimation as well as data selection strategy for estimation model training.

\subsection{Integration Quality Estimation}
\label{sec:estimation:iq}

According to Definition~\ref{def:iq}, the $IQ$ of a join path, \( P = T_{base} \bowtie_{0,1} T_1 \bowtie_{1,2} T_2 \cdots \bowtie_{n-1,n} T_n \), is determined by the number of instances in \( T_{base} \) that can be augmented by the instances in the feature column $c_{n,j}$ from $T_n$. This essentially boils down to the \textit{sequence of join operations} in a path. Inspired by cardinality estimation in database systems \cite{10.14778/3368289.3368296,10.14778/3551793.3551859}, we introduce a Long Short-Term Memory (LSTM)-based model designed specifically to estimate the integration quality of a given join path. The idea is to leverage the pairwise table features and statistics to estimate the $IQ$ without materializing the join. In the following, we introduce the input features extracted from join operations for LSTM network training.

\para{Transitivity.} Intuitively, the transitivity of a join between two tables, denoted as \( T_1 \bowtie_{1,2} T_2 \), indicates the number of \textit{unique} join keys in $T_2$ that can be found in $T_1$. We formally define it as \( \frac{\{T_{i}.c_{i,j}\}\ \cap\ \{T_{i+1}.c_{i+1,j}\}}{\{T_{i+1}.c_{i+1,j}\}} \), where $\{T_i.c_{i,j}\}$ denotes the set of keys in $T_i$ that can join with the set of keys in $T_{i+1}$. For instance, if the join keys $\{T_1.c_{1,2}\}$ are \{1,2,...,100\} and $\{T_2.c_{2,2}\}$ are \{81,82,...,120\}, then the transitivity between $T_1$ and $T_2$ is 50\%. And the transitivity of a join path can be aggregated by multiplying pairwise transitivity of each join operation. For a join path \( P = T_{base} \bowtie_{0,1} T_1 \bowtie_{1,2} T_2 \), where $Transitivity(T_{base} \bowtie_{0,1} T_1)=20\%$ and $Transitivity(T_1 \bowtie_{1,2} T_2) = 50\%$, the transitivity of $P$ is 20\% $\times$ 50\% = 10\%.

Transitivity is a critical feature when the join keys are uniformly and independently distributed. However, it may not be sufficient when the assumptions of uniformity and independence of join keys do not hold. To mitigate this limitation, we introduce additional features associated with join operations. 

\para{Variance.} The variance of a join \( T_{i} \bowtie_{i,j} T_j \) is defined as 
\[ Var(T_{ij}.c_k) = \frac{1}{|T_{ij}.c_k|}\sum_{l=1}^n (x_l-\mu)^2, \]
where $T_{ij}$ is the joined table, $c_k$ is the join column between $T_i$ and $T_j$, $x_l$ is the instance $c_{k}$, and \( \mu = \frac{1}{|T_{ij}.c_k|}\sum_{l=1}^n x_l \). Intuitively, a small variance implies a lack of diversity in join results, signifying potential bias that compromises transitivity in subsequent joins.

\para{Entropy.} The entropy of a join \( T_{i} \bowtie_{i,j} T_j \) is defined as 
\[ Ent(T_{ij}.c_k) = -\sum_{l=1}^\mathcal{X} P(x_l) \cdot \log(1-P(x_l)), \]
where $T_{ij}$ is the joined table, $c_k$ is the join column between $T_i$ and $T_j$, and $\mathcal{X}$ denotes the unique values in the join column $c_k$. While variance is suitable for numeric columns, entropy can serve as a complementary feature for textual and categorical columns. 

\para{KL-divergence.} This metric quantifies the column distribution difference before and after a join, defined as
\[ KL(T_j.c_k, T_{ij}.c'_k) = \sum_{x\in \mathcal{X}} P(x) \cdot \log(\frac{P(x)}{Q(x)}), \]
where $T_j.c_k$ is the column before join, $T_{ij}.c'_k$ is the column after join, and $\mathcal{X}$ denotes the unique values in a categorical column $T_j.c_k$ or the binning of a numeric column $T_j.c_k$. A large KL-divergence indicates potential bias introduced by the join.

\para{Pearson Correlation Coefficient.} This metric is a normalized measurement of the covariance, capturing the positive or negative linearity between two join columns. The Pearson Correlation Coefficient can be computed as
\[ PC(T_i.c_j, T_{i+1}.c_k) = \frac{cov(T_i.c_j)}{\sqrt{var(T_i.c_j)var(T_{i+1}.c_k)}}, \]
where $T_i.c_j$ and $T_{i+1}.c_k$ are two join columns in $T_i$ and $T_{i+1}$.

\para{Mutual Information.} Similar to Pearson Correlation Coefficient, mutual information quantifies the shared information between two columns, which can be computed as
\[ MI(T_i.c_j, T_{i+1}.c_k) = \sum_{x\in T_i.c_j}\sum_{y\in T_{i+1}.c_k} P(x,y)log(\frac{P(x,y)}{P(x)P(y)}), \]
where $T_i.c_j$ and $T_{i+1}.c_k$ are two join columns in $T_i$ and $T_{i+1}$, respectively. Mutual information evaluates both linear and non-linear relationships between two join columns.

\begin{figure}
  \centering
  \includegraphics[width=\columnwidth]{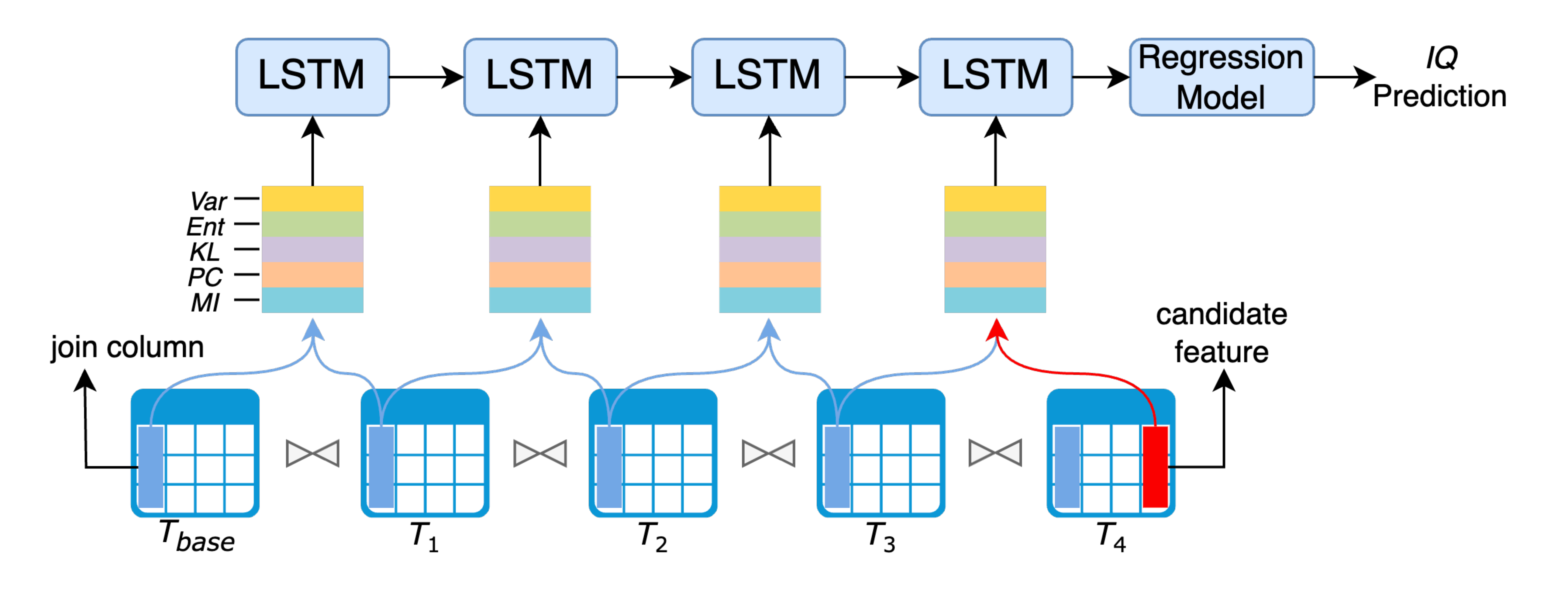}
  \caption{LSTM network for $IQ$ estimation.}
  \label{fig:lstm}
\end{figure}

We use the above feature characteristics of join operations to train a Long Short-Term Memory (LSTM) model for $IQ$ estimation. The LSTM architecture is chosen due to its lightweight yet powerful capability in capturing the inherent order and dependencies within a sequence of joins in a join path. As depicted in Figure~\ref{fig:lstm}, the features, namely $Var$, $Ent$, $KL$, $PC$ and $MI$, are concatenated together to represent each join operation \( T_i \bowtie_{i,j} T_j\). Note that we use the feature column $f_{i,j}$ instead of the join column in $T_4$ to compute the above features as it determines the ultimate $IQ$.

The training objective of the LSTM network using Mean Squared Error (MSE) as the loss function can be represented as follows:
\[ \min_{\theta} \frac{1}{N} \sum_{i=1}^{N} (y_i - \hat{y}_i(\theta))^2, \]
where \( N \) is the number of training examples, \( y_i \) is the true $IQ$ for the \( i \)-th join path sample, and \( \theta \) denotes the parameters of the LSTM model, and \( \hat{y}_i(\theta) \) is the predicted value by the model parameterized by \( \theta \). The training objective is to minimize the MSE loss by adjusting the weights of the LSTM. The method of selecting data points for training the LSTM network will be described in Section~\ref{sec:estimation:data}.

\subsection{Feature Importance Estimation}
\label{sec:estimation:fi}

Feature importance estimation has been studied~\cite{NIPS2017_8a20a862,hooker2018evaluating,vskrlj2020feature} in recent years. However, these methods are often computationally expensive and model specific. Hence we introduce a simple and effective clustering-based method to estimate the feature importance. The intuition is that features with similar representations in the embedding space tend to have a similar impact (i.e., feature importance) on the ML task~\cite{NIPS2017_8a20a862,metam}. The clustering algorithm exploits this observation to compute the representative features' $FI$ in each cluster and estimate $FI$ of other features based on their proximity to the representative ones, which leads to substantial computational savings. In the following, we explain four steps to estimate feature importance in detail.

\para{Feature Representation.} The representation of a feature column $c_{i,j}$ consists of two parts: (1) column metadata, and (2) column data instances. 
We use a pre-trained language model (PLM) such as BERT~\cite{devlin2018bert} to generate the embeddings of column metadata such as column names, column semantic types, column description, if available. For data instances in a feature column, we sample the column, serialize the sample data, and use a PLM to generate the embeddings as well. Finally, we concatenate the embeddings of both column metadata and data instances as the feature representation.

\para{Feature Clustering.} We design the clustering method based on DBSCAN~\cite{ester1996density}. To be specific, we first randomly select an unexplored feature and check the number of features that lie within its $\epsilon$-neighborhood, where $\epsilon$ denotes the threshold of cosine similarity between two feature representations. Then we check if the $\epsilon$-neighborhood of this selected feature contains at least $m$ features. 
If so we form a cluster by adding all unexplored features within the selected feature's $\epsilon$-neighborhood. We iterate this process until all features have been assigned to a cluster. Unlike DBSCAN, we do not consider any feature as noise. Namely, the remaining unexplored features will be considered as clusters of their own. 
The advantage of this lightweight method is that it bounds the similarity between the centroid feature and other candidate features within a cluster, i.e., \( cos(Repr(centroid), Repr(candidate_i)) \leq \epsilon \). 

\para{Compute $FI$.} According to Definition~\ref{def.fi}, we need to follow one or more join paths to integrate a feature to the base table in order to compute its $FI$. Naively, one could utilize all join paths between the base table and the feature table to obtain its maximum $FI$, which is very costly. In Section~\ref{sec:estimation:data}, we will introduce a data selection strategy that chooses features and their associated join paths for both $FI$ computation and LSTM network training.

\para{Estimate $FI$ for Unseen Features.} For a unseen feature, we estimate its $FI$ by a weighted average over the $FI$ of its surrounding features in the same feature cluster. Precisely, we have
\[ FI(c) = \frac{1}{|N|}\sum_{n=1}^N FI(c_n) \cdot cos(c, c_n), \]
where $c$ denotes the unseen feature, $N$ denotes the set of features with known $FI$ in the same cluster, and $cos(\cdot)$ represents the cosine similarity between the unseen feature and its neighboring feature. For example, Figure~\ref{fig:cluster} illustrates $f_{1,2}$'s $\epsilon$-neighborhood. The distances between the features are calculated using cosine similarity over their feature representations. The $FI$ scores of the feature represented in orange circles $f_{1,2}$, $f_{2,4}$ and $f_{3,6}$ are known. The $FI$ scores of the unseen features $f_{4,2}$ and $f_{5,3}$ (green circles) are estimated by referring to the three explored features.

\begin{figure}
  \centering
  \includegraphics[width=0.8\columnwidth]{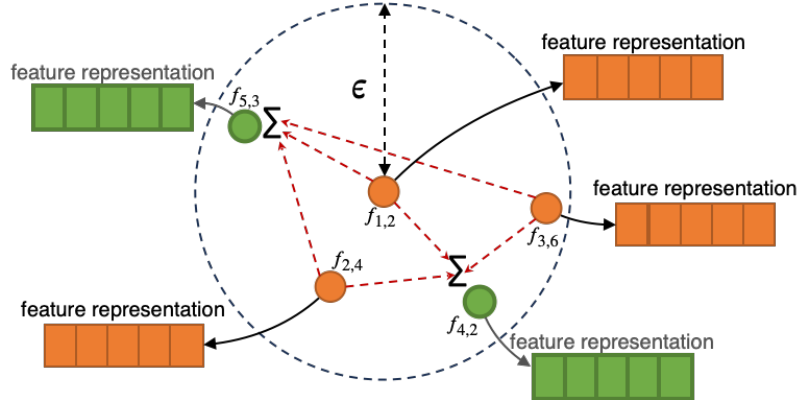}
  \caption{Clustering-based $FI$ estimation.}
  \label{fig:cluster}
\end{figure}

\subsection{Data Selection for Estimation}
\label{sec:estimation:data}

Now we introduce an effective data selection strategy for $IQ$ model training and $FI$ estimation, both of which require expensive join executions according to the chosen join paths. In addition, $FI$ needs to feed the augmented table to an AutoML framework (e.g., AutoGluon) to obtain the actual feature importance. Hence our goal is to choose high quality data points to serve both $IQ$ and $FI$ estimations, in the form of \textit{feature-path} ($f_{i,j}$, $P_k$) pairs. Ideally, the selected data points should have two desired properties. First, the selected integration paths should have a good coverage over all joins among any candidate tables and the base table. The covered joins can be stitched together to accurately estimate unexplored integration paths. Second, the selected features should be evenly distributed across all feature clusters. The number of features chosen from one cluster should be proportioned to the size of the cluster.

To achieve two properties above, we consider the following factors regarding a feature-path pair and propose a weight function to rank all feature-path pairs.

\begin{enumerate}[leftmargin=*]
    \item Single join path coverage ($PC$) represents how much a join path has been explored. This is measured by the ratio between the number of join edges in a join path that have been explored and the total number of join edges in the path.
    \item $l$-hop join path group coverage ($PGC$) represents the number of join paths of a length $l$ that have been fully explored. This is measured by the ratio between the number of fully explored $l$-hop join paths and the total number of $l$-hop paths.
    \item Feature cluster coverage ($FC$) represents the number of features that have been examined in a feature cluster. This is measured by the ratio between the number of tested features and the total number of features in a cluster.
\end{enumerate}


Note that all the above factors range between [0, 1] and a lower value of each factor indicates that a path or a feature should be explored. Hence the weight function is defined as
\begin{equation}
w(f_{i,j}, P_k) = PC \times PGC \times FC.
\label{eq:weight}
\end{equation}

We assign a score to each feature-path pair using this weight function, rank them in an ascending order to select the one with the smallest weight in each iteration. After each iteration, we update the weights, re-rank the feature-path pairs and continue the selection until a given budget is exhausted. Algorithm~\ref{algo.datapointselection} details the greedy data point selection method.


\begin{algorithm}[t]
\caption{Greedy Data Point Selection}
\label{algo.datapointselection}
\footnotesize
\begin{algorithmic}[1] 
\Require All feature-path pairs (\textit{candidateSet}), the max number of paths to select (\textit{budget})
\Ensure The set of selected feature-paths (\textit{selected})
\Function {data\_point\_selection}{}
\State $unseleted \gets candidateSet$
\State $selected \gets \emptyset$
\State $path\_coverage \gets$ $init\_pc(unseleted$)
\State $path\_group\_coverage \gets$ $init\_pgc(unseleted$)
\State $cluster\_coverage \gets$ $init\_fc(unseleted$)
\State $weights \gets$ $init\_weight(unselected)$
\While{$budget$ > 0 \& $unselected\neq \emptyset$}
    \State $(f, p) \gets min\_weight(unselected, w)$
    \State $selected \gets selected \cup \{(f, p)\}$ 
    \State $unselected \gets unselected - \{(f, p)\} $
    \For{$e \in p$}
        \State $path\_coverage[e] \gets True$
    \EndFor
    \State $path\_group\_coverage[len(p)]$ += 1
    \State $cluster\_coverage[f.getCluster()]$ += 1
    \State $weights \gets$ $init\_weight(unselected)$
    \For{$(f,p) \in unselected$}
        \State $weights[(f,p)] \gets w(f, p, path\_coverage, \newline\hspace*{6em} path\_group\_coverage, cluster\_coverage)$
    \EndFor
    \State $budget \gets budget-1$
\EndWhile
\State \Return{$selected$}
\EndFunction        
\end{algorithmic}
\end{algorithm}

Algorithm~\ref{algo.datapointselection} first initializes several data structures including the ones to store the intermediate states of single join path coverage, $l$-hop join path group coverage, feature cluster coverage, and the weights for all unselected feature-path pairs (Lines 1-7). Then the feature-path pair with the smallest weight is chosen and removed from the set of unselected pairs (Lines 9-11). Note that we randomly select one feature-path pair if there are multiple ones with the smallest weight. Once a feature-path pair is chosen, the join path coverage, join path group coverage, and feature cluster coverage are updated accordingly (Lines 12-17). The weights of the remaining unselected pairs are recomputed based on Equation~\ref{eq:weight} (Lines 18-20). We iterate this process until the given \textit{budget} is exhausted or all feature-path pairs are selected. 


\section{Feature Search and Integration}
\label{sec:exploit}

In this section, we present \method's search algorithm that efficiently and effectively exploits the integration quality and feature importance estimations to identify an augmentation plan that maximizes the utility gain of a ML task.

\subsection{Join Graph Model}
\label{sec:exploit:graph}

We first introduce \method join graph model to capture the search space of all integration plans.

\begin{definition}[\method Join Graph]
A \method join graph $\mathcal{G} = (V, E)$ is a unweighted and undirected graph with a set of vertices and a set of edges. A vertex $v \in V$ represents a table, consisting of a set of features $F = \{f_1,...,f_n\}$, and an edge $e \in E$ represents a join between two tables.
\end{definition}



Given a join graph $\mathcal{G}$, a source vertex $v_b$ (i.e., the base table), and a given budget $B$, the goal is to identify no more than $B$ features and their respective integration paths in $\mathcal{G}$ to maximize the utility score of the ML task. Note that each path starts from $v_b$ and ends with a feature $f_{i,j}$ associated with a vertex $v_i$.


Naively, we can follow the breadth-first search (BFS) approach to enumerate all integration paths from $v_b$. For a path $P_j$ leading to $v_i$, we form feature-path pairs for all features in $v_i$. Among all feature-path pairs, we generate all possible candidate combinations with $\leq B$ features. The utility score of each candidate combination is computed based on the $IQ$ and $FI$ estimation. The combination with the maximum utility score represents the optimal solution to the feature augmentation problem. While this naive solution guarantees an optimal solution, it suffers from the combinatorial explosion problem. Specifically, real-world datasets often consist of many candidate tables and join relationships, resulting in a dense join graph. For instance, the \textsf{School} dataset from our experiments has over 10,000 join paths. Worse yet, with $n$ candidate features, a budget of at most selecting $B$ of them, and on average $m$ integration paths for each feature, the number of combinations is \( \sum_{i=1}^B m\cdot\binom{n}{i} \) (e.g., the \textsf{School} dataset has more than $10^8$ combinations).

\subsection{Search and Refinement}
\label{sec:exploit:search}

We introduce two pruning principles below to progressively reduce the search space, preventing the search explosion in the early stage without sacrificing the quality of feature augmentation.

\para{Pruning Principle 1 ($IQ$ Monotonicity).} Based on Definition~\ref{def:iq}, the integration quality monotonically decreases as a path extends. In other words, if the $IQ$ of a path $P$ is too low, it is not beneficial to explore any path that has $P$ as its prefix. This fundamental principle is critical for path exploration with early termination.

\para{Pruning Principle 2 ($FI$ Lower Bound).} If the feature importance of a feature $f_i$ is too low, then $f_i$ and all integration paths to $f_i$ can be safely pruned. The rationale behind this principle is based on Theorem~\ref{th:fi&iq}. Specifically, if $f_i$'s $FI$ is too low, its $UG$ would still be negligible even there were a perfect integration path (i.e., $IQ$ = 1). Note that the lower bound of $FI$ can be progressively increased as we explore the search space and identify more useful features.

\para{Search with Progressive Pruning.}
Leveraging these two pruning principles, we introduce \method's search algorithm (Algorithm~\ref{algo:search}) for feature selection and their corresponding join path exploration. Algorithm~\ref{algo:search} employs a BFS-based approach to identify feature candidates and their integration paths iteratively. It starts with the base table (Line 2) and iteratively extends its search to $k$-hop features (Lines 4-25). 

During this process, we maintain a min-heap of size $H$ that keeps top-$H$ candidate feature-path pairs $candidates$, sorted by the $UG$ score in a descending order (Line 3). Note that the heap size $H$ is larger than the feature budget $B$ so as to maintain a superset of candidates for further refinement. In practice, we set $H$ to a multiple of $B$. We then iterate through each $k$-hop path in $P_{cur}$ (line 6), extending the path to include an additional vertex $v$ directly connected to the last vertex $u$ of the path (Lines 6-12). Subsequently, we compute the $IQ$ for the extended path ($p'$) and discard it if its $IQ$ is below $T_{IQ}$ (Lines 13-14). For the qualified $p'$, we add it to the set of $(k+1)$-hop paths $P_{next}$, and evaluate the $UG$ for each feature on vertex $v$. The feature-path pair with the lowest $UG$ score in $candidates$ is used as the lower bound. We update $candidates$ with the qualified feature-path pair (Lines 16-18). Finally, the algorithm returns the candidate set (Line 26).

\begin{algorithm}[t]
\caption{Search Algorithm}
\label{algo:search}
\footnotesize
\begin{algorithmic}[1] 
\Require Join graph ($G=(E,V)$), $IQ$ threshold ($T_{IQ}$), Heap size $H$
\Ensure Feature-path candidates ($candidates$)
\Function {search}{}
\State $P_{cur} \gets \{(v_b)\}$
\State $candidates=minHeap(H)$ \Comment{The min-heap size of $candidates$ is $H$}
\While{$|P_{cur}|>0$}
    \State $P_{next} \gets \emptyset$
    \For{$p \in P_{cur}$}
        \State $u \gets$ $p.getLastV()$ \Comment{Get the last vertex of path $p$}
        \For {$v \in u.getNeighbor()$} \Comment{Get one-hop neighbors of $v$} 
            \If{$v \in p$}
            \State \textit{continue} \Comment{Avoid cyclic joins}
            \EndIf
            \State $p' \gets p \oplus v$
            \State $iq \gets estimateIQ(p')$
            \If {$iq \geq T_{IQ}$}
            \State {$P_{next} \gets P_{next} \cup p'$}
            \For{$f \in v.getFeatures()$}
                \If{$estimateFI(f)\times iq \geq candidates.minUG()$} 
                \State $candidates.insert((f,p'))$
                \EndIf
            \EndFor
            \EndIf
        \EndFor
    \EndFor
    \State $P_{cur} \gets P_{next}$
\EndWhile
\State \Return {$candidates$}
\EndFunction    
\end{algorithmic}
\end{algorithm}

\begin{algorithm}[t]
\caption{Refinement Algorithm}
\label{algo:augmentation}
\footnotesize
\begin{algorithmic}[1] 
\Require Feature-path pairs $candidates$, Feature budget $B$
\Ensure Feature augmentation plan $selected$
\Function {refinement}{}
\For {$c \in candidates$}
    \State $c.ug \gets c.fi\cdot c.iq$
\EndFor
\State $unselected \gets candidates$
\State $selected \gets \emptyset$
\While{$B>0$}
    \State $B \gets B-1$
    \State $fp_{max} \gets unselected.getMaxUG()$
    \State $selected \gets selected \cup fp_{max}$
    \State $unselected \gets unselected - fp_{max}$
    \For {$c \in unselected$}
        \State $c.ug \gets$ $updateUG(c, selected)$
    \EndFor
\EndWhile
\State \Return {$selected$}
\EndFunction
\end{algorithmic}
\end{algorithm}

\para{Feature Augmentation Refinement.}
Among the selected feature-path candidates, similar features in the same cluster are expected to have comparable effect on the downstream ML task~\cite{metam}. In light of this observation, we design a weight decay method to refine the selected feature candidates and to produce the final feature augmentation plan. In essence, this weight decay method aims to mitigate the local optimum incurred by choosing similar features during the search process. 

Specifically, the weight decay method relies on a scaling factor using feature similarity to regulate the $FI$ on related features, subsequently optimizing the overall efficacy of selected features in the downstream tasks. The weight decay function on a feature $f$ is
\[ FI(f) = \frac{1}{|N|}\sum_{i=1}^N FI(f_i) \cdot (1 - cos(f, f_i)), \]
where $N$ represents the set of already chosen features in the same cluster. 
The closer a feature to the other chosen features in the same cluster, the more penalty it receives. In the refinement process, we iteratively choose the feature-path pair with the maximum $UG$ and update the $UG$ of the remaining pairs based on the equation above. After $B$ iterations, we select top-$B$ feature-path pairs as the feature augmentation plan.

Algorithm~\ref{algo:augmentation} illustrates the refinement strategy in detail. It first initializes the estimated $UG$ based on $IQ$ and $FI$ of each candidate feature-path pair (Lines 2-4). During each iteration, it performs two key steps: (1) among unselected feature-path pairs, it selects the feature and corresponding path with the highest estimated $UG$ (Line 9); and (2) it updates the estimated $UG$ for remaining feature-path pairs based on the above equation (Line 13). This adjustment indicates that a higher similarity between a unselected feature and the already chosen ones leads to a lower $UG$ value for the unselected feature, implying a reduced contribution to the ML task. Finally, the selected feature-path pairs $selected$ are returned after a given feature budget $B$ is exhausted (Line 16).

\section{Experimental Evaluation}
\label{sec:exp}

\subsection{Experimental Setup}
\label{sec:exp:setup}

\para{Datasets.} We evaluate \method over five public datasets by performing classification and regression tasks. All five datasets are composed of base tables provided by Kaggle and the DARPA Data Driven Discovery of Models (D3M)\footnote{\url{https://datadrivendiscovery.org/about-d3m/}}. Given a base table, we searched open source datasets for joinable tables using NYU Auctus~\cite{auctus} to identify candidate tables for augmentation purpose.
\begin{itemize}[leftmargin=*]
    \item \textbf{School} is a dataset commonly used in baselines~\cite{arda,autof,metam} for a classification task. The target prediction is school performance on a standardized test based on student attributes, course attributes and other historical surveys.
    \item \textbf{DonorsChoose\footnote{\url{https://www.kaggle.com/c/donorschoose-application-screening}}} is a Kaggle dataset for a classification task. The target prediction is whether or not a DonorsChoose proposal was accepted. We identified 122 candidate features for augmentation.
    \item \textbf{Fraud Detection\footnote{\url{https://www.kaggle.com/competitions/ieee-fraud-detection}}} is a Kaggle dataset  consisting of over 800 candidate features. The task is to predict whether an online transaction is fraudulent.
    \item \textbf{Poverty} consists of socioeconomic features like poverty rates, population change, unemployment rates, and education levels across U.S. States and counties. The task is to predict an index indicating poverty level.
    \item \textbf{Air} is a regression dataset from NYU Auctus and google aiming to predict the air quality of a city on a given date.
\end{itemize}
Note that tables from Kaggle are typically wide, consisting of many features already integrated with base tables. Hence we normalize these wide tables and consequently generate candidate tables with features more than 1-hop away from the original base tables. Table~\ref{tab:dataset_stat} provides the basic information about the three datasets and the numbers of immediate (1-hop) and distant features.

\begin{table}[t]
\centering
\resizebox{1\linewidth}{!}{\begin{tabular}{@{}l|c|c|c|c@{}}
\toprule
\textbf{Datasets} & \textbf{\#Tables} & \textbf{\#Columns} & \textbf{\#1-hop Features} & \textbf{\#Distant Features} \\ \midrule
\textsf{School}          & 121 & 1,295 & 440 & 625 \\
\textsf{DonorsChoose}    & 73  & 221   & 48  & 74  \\ 
\textsf{Fraud Detection} & 81  & 254   & 30  & 97  \\
\textsf{Poverty} & 98 & 408 & 93  & 117 \\
\textsf{Air} & 75  & 603 & 172 & 308 \\
\bottomrule
\end{tabular}}
\caption{Statistics of datasets.}
\vspace{-20pt}
\label{tab:dataset_stat}
\end{table}

\begin{table*}[htbp]
\centering
\begin{tabular}{c|c|c|c|c|c|c|c}
\toprule
\multirow{2}{*}{\textbf{Datasets}} & \multirow{2}{*}{\textbf{Metric}} & \multirow{2}{*}{\textbf{Feature Budget}} & \multicolumn{5}{c}{\textbf{Methods}} \\
\cline{4-8}
&  &  & DFS & ARDA & AutoFeature & METAM & \method \\
\hline
\hline
\multirow{3}{*}{\textsf{School}} & \multirow{3}{*}{Accuracy} & 1 & 0.704{\scriptsize(4)} & 0.697{\scriptsize(5)} & 0.708{\scriptsize(3)} & 0.790{\scriptsize(2)} & \textbf{0.823}{\scriptsize(1)} \\
& & 5 & 0.700{\scriptsize(5)} & 0.808{\scriptsize(2)} & 0.704{\scriptsize(4)} & 0.801{\scriptsize(3)} & \textbf{0.891}{\scriptsize(1)} \\
& & 10 & 0.692{\scriptsize(5)} & 0.794{\scriptsize(3)} & 0.723{\scriptsize(4)} & 0.816{\scriptsize(2)} & \textbf{0.880}{\scriptsize(1)} \\
\hline
\multirow{3}{*}{\textsf{DonorsChoose}} & \multirow{3}{*}{Accuracy} & 1 & 0.655{\scriptsize(5)} 
& \textbf{0.856}{\scriptsize(1)} & 0.708{\scriptsize(3)} & 0.656{\scriptsize(4)} & 0.822{\scriptsize(2)} \\
& & 5 & 0.820{\scriptsize(4)} & 0.890{\scriptsize(2)} & 0.852{\scriptsize(3)} & 0.659{\scriptsize(5)} & \textbf{0.954}{\scriptsize(1)} \\
& & 10 & 0.854{\scriptsize(4)} & 0.901{\scriptsize(2)} & {0.896}{\scriptsize(3)}
& 0.820{\scriptsize(5)} &\textbf{0.961}{\scriptsize(1)} \\
\hline
\multirow{3}{*}{\textsf{Fraud Detection}} & \multirow{3}{*}{F1} & 1 & 0.068{\scriptsize(5)} & 0.416{\scriptsize(2)} & 0.145{\scriptsize(4)} & \textbf{0.437}{\scriptsize(1)} & 0.435{\scriptsize(3)} \\
& & 5 & 0.070{\scriptsize(5)} & 0.422{\scriptsize(3)} & 0.152{\scriptsize(4)} & 0.446{\scriptsize(2)} & \textbf{0.493}{\scriptsize(1)} \\
& & 10 & 0.084{\scriptsize(5)} & 0.450{\scriptsize(3)} & 0.162{\scriptsize(4)} & 0.464{\scriptsize(2)} & \textbf{0.540}{\scriptsize(1)} \\
\hline
\hline
\multirow{3}{*}{\textsf{Poverty}} & \multirow{3}{*}{MAE} & 1 & 13620.14 {\scriptsize(5)}  & 12389.54 {\scriptsize(2)}  & 13532.57 {\scriptsize(4)}   & 13077.66 {\scriptsize(3)}  & \textbf{8222.34} {\scriptsize(1)}  \\
& & 5 & 13410.07 {\scriptsize(4)}  & 12389.54 {\scriptsize(2)}  & 13532.57 {\scriptsize(5)}  & 12956.29 {\scriptsize(3)}  & \textbf{7322.44} {\scriptsize(1)}  \\
& & 10 & 13077.66 {\scriptsize(4)}  & 12164.23 {\scriptsize(2)}  & 13411.85 {\scriptsize(5)}  & 12786.82 {\scriptsize(3)}  & \textbf{7182.38} {\scriptsize(1)}  \\
\hline
\multirow{3}{*}{\textsf{Air}} & \multirow{3}{*}{MSE} & 1 & 1.184 {\scriptsize(4)} & \textbf{0.969} {\scriptsize(1)}  & 1.259 {\scriptsize(5)}  & 1.101 {\scriptsize(2)}  & 1.101 {\scriptsize(2)}  \\
& & 5 & 0.985 {\scriptsize(4)}  & 0.793 {\scriptsize(2)}  & 1.219 {\scriptsize(5)} & 0.900 {\scriptsize(3)}  & \textbf{0.762} {\scriptsize(1)}  \\
& & 10 & 0.943 {\scriptsize(4)} & 0.761 {\scriptsize(2)}  & 1.202 {\scriptsize(5)}  & 0.762 {\scriptsize(3)}  & \textbf{0.715} {\scriptsize(1)}  \\
\bottomrule
\end{tabular}
\caption{Baseline comparison against the state-of-the-art methods. The top ranked results are bold. The numbers in the parentheses indicate the result ranking on a particular dataset.}
\label{tab.main_exp}
\end{table*}

\para{Baselines.} We compared \method with a variety of baselines, including heuristic-based solutions such as Deep Feature Synthesis and ARDA as well as machine learning-based solutions such as AutoFeature and METAM.

\begin{itemize}[leftmargin=*]
    \item Deep Feature Synthesis (DFS)~\cite{dfs} is a pioneer work that automatically generates features for relational datasets. It utilizes join paths in the data to a base table, and then applies mathematical functions along these paths to create the final features.
    \item ARDA~\cite{arda} is a feature augmentation system, which employs a variety of heuristics to select top-$k$ tables to join and uses a random injection-based feature augmentation to search candidate feature subsets.
    \item AutoFeature~\cite{autof} is a reinforcement learning based framework to augment the features following an exploration-exploitation strategy over the search space of candidate tables (features).
    \item METAM~\cite{metam} is a goal-oriented framework that queries the downstream task with a candidate dataset, using a feedback loop that automatically steers the discovery and augmentation process.
\end{itemize}

\para{Evaluation Metrics.} Since we evaluate \method using both classification and regression tasks, we adopt the commonly used metrics for measuring the effectiveness, namely:
\[ Accruacy = \frac{\#\, of\, correct\, predictions}{\#\, of\, predictions}, \text{and} \]
\[ F1 = \frac{TP}{TP+\frac{1}{2}(FP+FN)}, \]
where $TP$, $FP$ and $FN$ denote the number of \textit{true positive}, \textit{false positive} and \textit{false negative} predictions, respectively. In particular, we report Accuracy for the classification tasks on \textsf{DonorChoose} and \textsf{School} because of the balanced label distribution over the classes. We report F1 for the classification task on \textsf{Fraud Detection} since the label distribution is skew towards the negative classes and typically a fraudulent (positive) case is a more valuable prediction. For the regression tasks, we follow ~\cite{arda,autof} and report the Mean Absolute Error (MAE) or Mean Square Error (MSE) between the estimated and ground truth values:
\[ MAE = \frac{\sum^{n}_{i=1}|y_i-x_i|}{n} \text{ and } MSE = \frac{\sum^{n}_{i=1}(y_i-x_i)^2}{n}, \]
where $n$ is the number of estimated values, $y_i$ is the ground truth value for the $i$-th estimated value $x_i$. In ablation studies, we also measure the accuracy of each \method's component by reporting the MAE or MSE. Note that we use AutoGluon to automatically train multiple models (e.g., XGBoost model, KNN model, neural network models, etc.), and then ensemble the models to create the final predictor.


\subsection{Baseline Comparison}
\label{sec:exp:main}

We first conduct an end-to-end evaluation of \method against the state-of-the-art solutions for feature augmentation. In this set of experiments, we study the effectiveness of \method on finding the most useful features from candidate tables within a given budget. The budge $B$ is defined as the number of features allowed to augment the base table. We vary the budget from 1, 5 to 10 and report the utility measurements of different methods on the datasets. As shown in Table~\ref{tab.main_exp}, when $B>$ 1, \method consistently demonstrates superior performance on finding and integrating features that contribute the most to the task on the base table. The reason is that \method not only accurately identifies useful features that are different hops away but also exploits these features with high-quality integration paths. 

When the budget is extremely limited (i.e., $B$ = 1), \method is still among the top 2 performing methods. In both \textsf{DonorChoose} and \textsf{Air} datasets, the most useful feature is from an 1-hop candidate table, which also presents high joinability to the base table. And this information is given to ARDA and hence it is able to outperform \method. METAM on the \textsf{Fraud Detection} dataset delivers the best performance by finding a 1-hop join path to the most useful feature. \method is slightly worse than METAM due to its $IQ$ and $FI$ estimation errors amplified by such an extreme budget.

\vspace{-7pt}
\subsection{Ablation Study}
\label{sec:exp:ablation}

To better understand how our key designs of \method benefit the overall performance feature augmentation, we conduct extensive ablation studies to evaluate the performance of \method's individual components and their variants.

\vspace{-3pt}
\subsubsection{Effectiveness of \method's Integration Quality Model.}

First, we evaluate the accuracy of \method's LSTM-based $IQ$ estimator. The ground truth is established by materializing the join paths and computing the corresponding $IQ$ scores according to Definition~\ref{def:iq}. We measure and report the MSE between the estimated values and the ground truth values. 
To make sense of the accuracy of the proposed $IQ$ estimator, we compare it against two baselines: (1) multiplication of \textit{Transitivity} scores between the tables on a join path (named Production), and (2) randomly predicting an $IQ$ score between 0 and 1 (named Random [0,1]). In addition, we include an LSTM-based $IQ$ estimator trained on randomly selected feature-path pairs and vary the number of training pairs.

\begin{figure*}
  \centering  
    \begin{subfigure}[b]{0.8\textwidth}
    \centering
    \includegraphics[width=\textwidth]{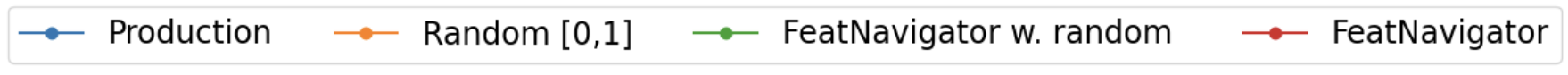}
    \end{subfigure}
    \hfill
    
  \begin{subfigure}[b]{0.19\textwidth}
    \centering
    \includegraphics[width=\textwidth]{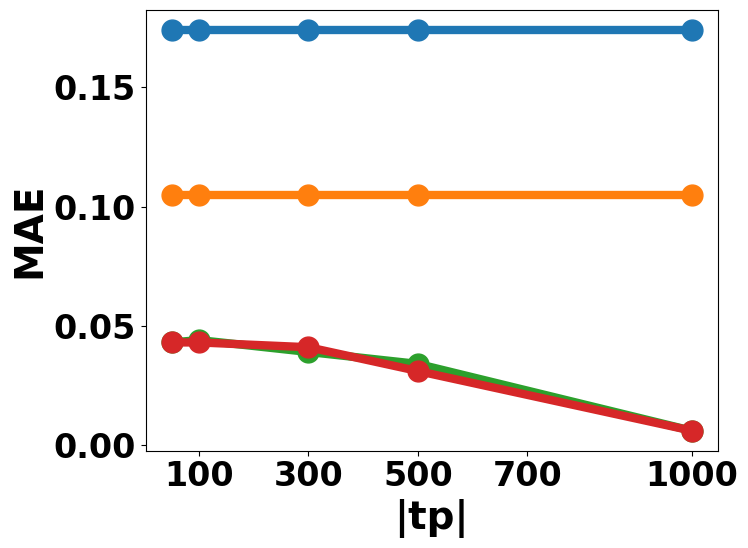}
    \caption{School}
    \label{fig:iq:school}
  \end{subfigure}
  \hfill
  \begin{subfigure}[b]{0.19\textwidth}
    \centering
    \includegraphics[width=\textwidth]{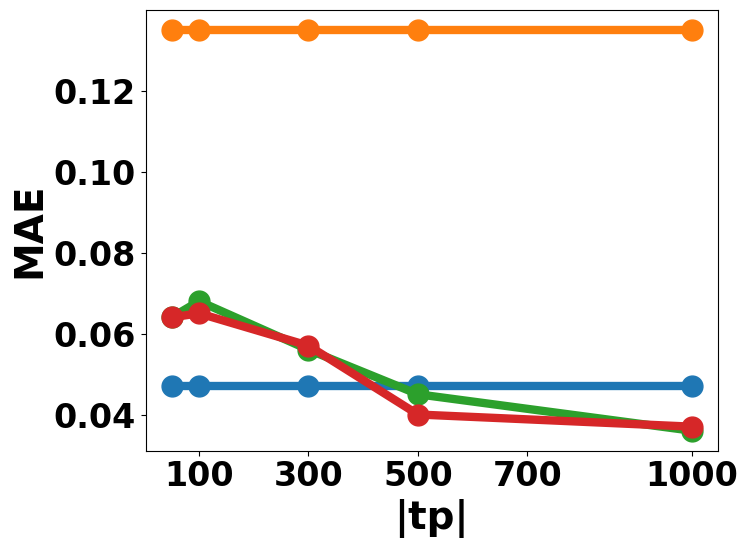}
    \caption{DonorsChoose}
    \label{fig:iq:donor}
  \end{subfigure}
  \hfill
  \begin{subfigure}[b]{0.19\textwidth}
    \centering
    \includegraphics[width=\textwidth]{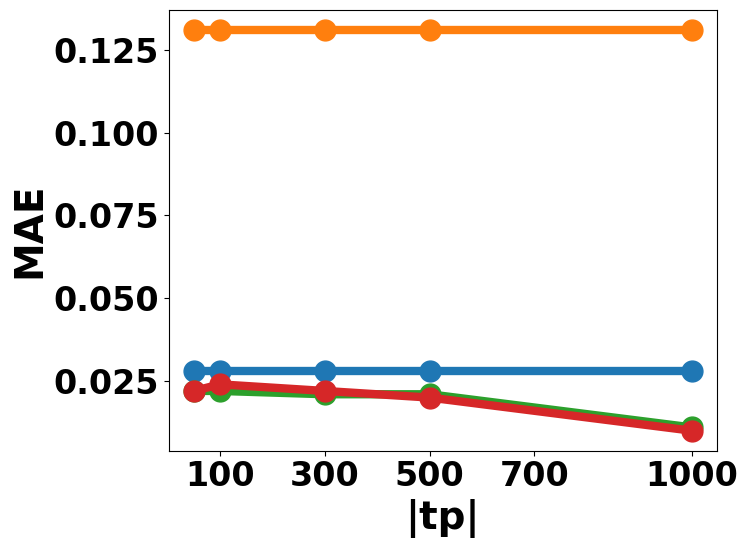}
    \caption{Fraud Detection}
    \label{fig:iq:transaction}
  \end{subfigure}
  \hfill
  \begin{subfigure}[b]{0.19\textwidth}
    \centering
    \includegraphics[width=\textwidth]{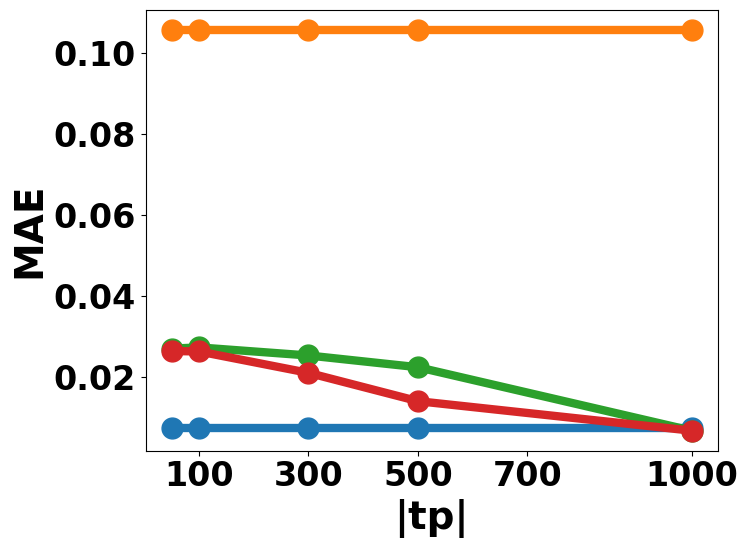}
    \caption{Poverty}
    \label{fig:iq:poverty}
  \end{subfigure}
  \hfill
  \begin{subfigure}[b]{0.19\textwidth}
    \centering
    \includegraphics[width=\textwidth]{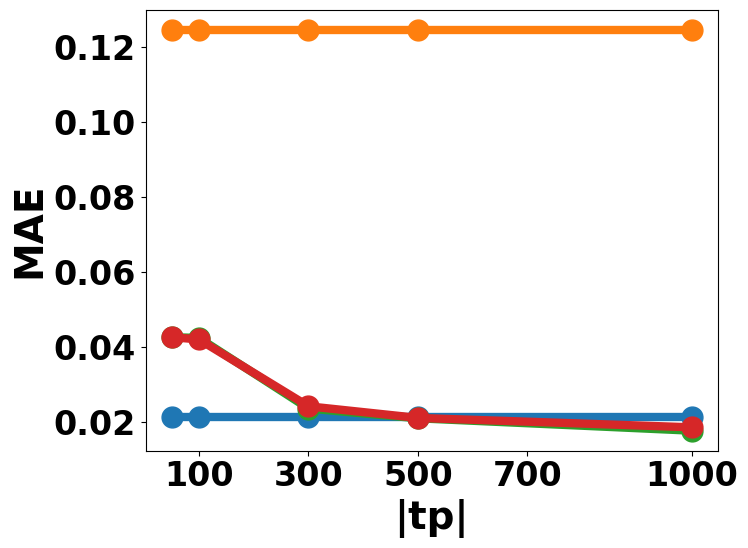}
    \caption{Air}
    \label{fig:iq:air}
  \end{subfigure}  
  \caption{Effectiveness of \method's integration quality model.}
  \label{fig:iq_budget_exp}
\end{figure*}
\begin{figure*}
  \centering
  \begin{subfigure}[b]{0.8\textwidth}
    \centering
    \includegraphics[width=\textwidth]{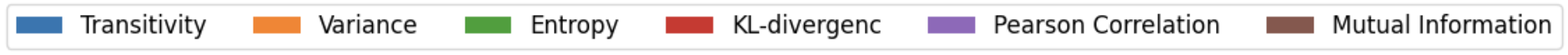}
  \end{subfigure}
  \hfill

  \begin{subfigure}[b]{0.19\textwidth}
    \centering
    \includegraphics[width=\textwidth]{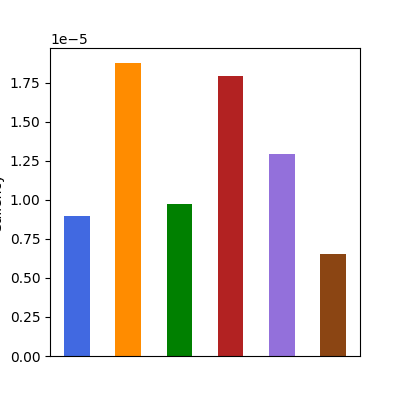}
    \caption{School}
    \label{fig:school}
  \end{subfigure}
  \hfill
  \begin{subfigure}[b]{0.19\textwidth}
    \centering
    \includegraphics[width=\textwidth]{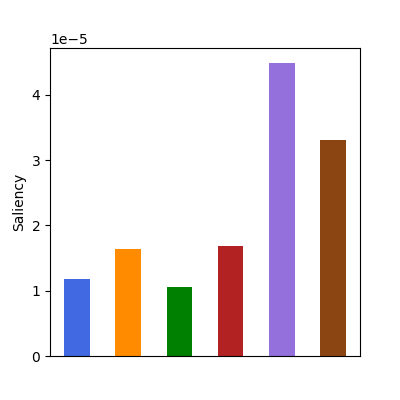}
    \caption{DonorsChoose}
    \label{fig:donor}
  \end{subfigure}
  \hfill
  \begin{subfigure}[b]{0.19\textwidth}
    \centering
    \includegraphics[width=\textwidth]{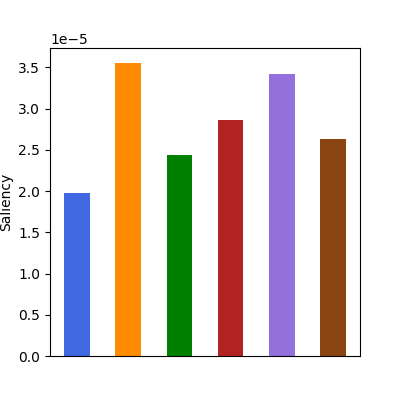}
    \caption{Fraud Detection}
    \label{fig:transaction}
  \end{subfigure}
  \hfill
  \begin{subfigure}[b]{0.19\textwidth}
    \centering
    \includegraphics[width=\textwidth]{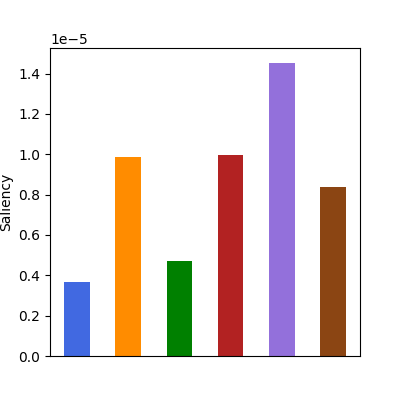}
    \caption{Poverty}
    \label{fig:poverty}
  \end{subfigure}
  \hfill
  \begin{subfigure}[b]{0.19\textwidth}
    \centering
    \includegraphics[width=\textwidth]{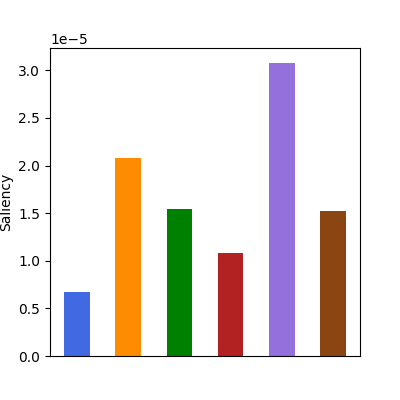}
    \caption{Air}
    \label{fig:air}
  \end{subfigure}
  \caption{Saliency analysis of features used in integration quality model.}
  \label{fig:saliency}
\end{figure*}

\begin{figure*}
  \centering
  \begin{subfigure}[b]{0.8\textwidth}
    \centering
    \includegraphics[width=\textwidth]{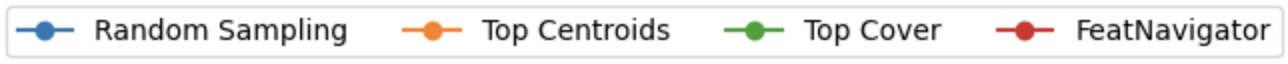}
  \end{subfigure}
  \hfill

  \begin{subfigure}[b]{0.19\textwidth}
    \centering
    \includegraphics[width=\textwidth]{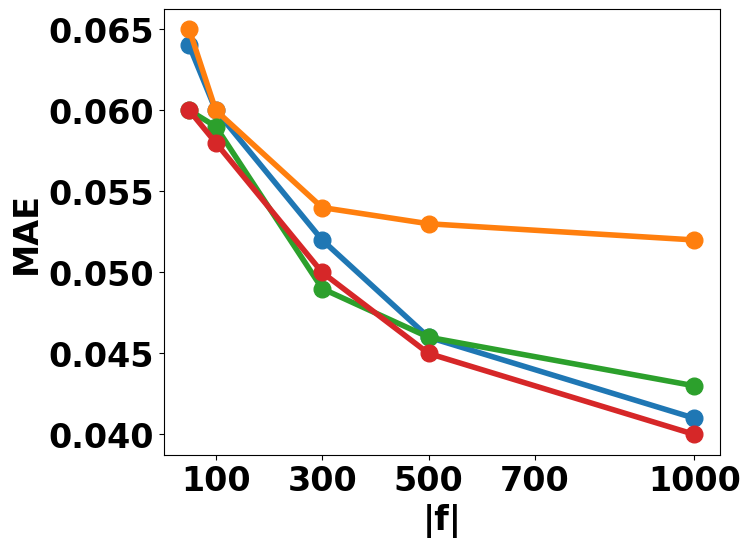}
    \caption{School}
    \label{fig:school}
  \end{subfigure}
  \hfill
  \begin{subfigure}[b]{0.19\textwidth}
    \centering
    \includegraphics[width=\textwidth]{figures/abla_iq/donor}
    \caption{DonorsChoose}
    \label{fig:donor}
  \end{subfigure}
  \hfill
  \begin{subfigure}[b]{0.19\textwidth}
    \centering
    \includegraphics[width=\textwidth]{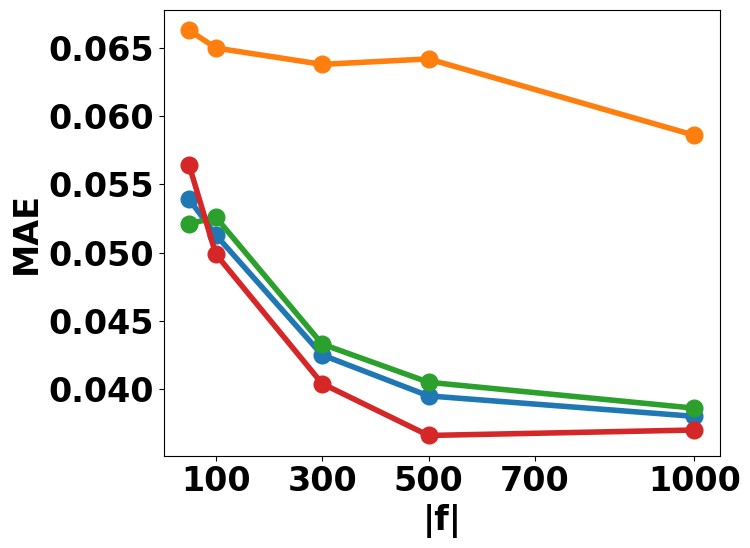}
    \caption{Fraud Detection}
    \label{fig:transaction}
  \end{subfigure}
  \hfill
  \begin{subfigure}[b]{0.19\textwidth}
    \centering
    \includegraphics[width=\textwidth]{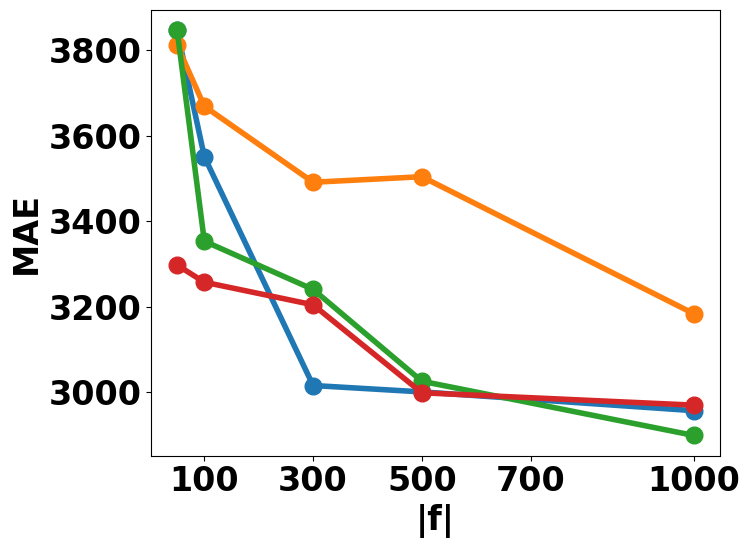}
    \caption{Poverty}
    \label{fig:poverty}
  \end{subfigure}
  \hfill
  \begin{subfigure}[b]{0.19\textwidth}
    \centering
    \includegraphics[width=\textwidth]{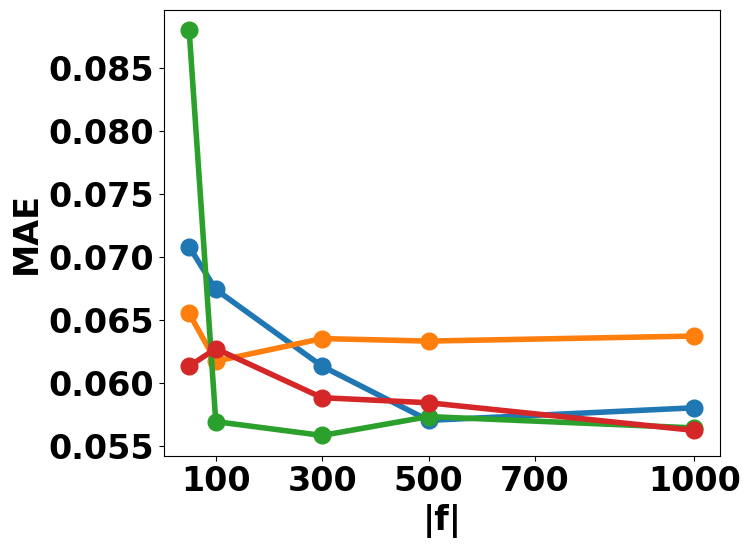}
    \caption{Air}
    \label{fig:air}
  \end{subfigure}
  \caption{Effectiveness of \method’s feature importance estimation.}
  \label{fig:fi_exp}
\end{figure*}

Overall, \method's LSTM-based $IQ$ estimator achieves the best (lowest) MSE compared to all baselines across three different datasets, as shown in Figure~\ref{fig:iq_budget_exp}. The random estimator performs the worst (highest) MSE. Notably, our method consistently outperforms both Production and Random methods, which empirically verifies the usefulness of the proposed features for our $IQ$ estimator. With the increasing number (from 50 to 1000) of the training join paths $|tp|$, our $IQ$ estimator produces more accurate estimations. Note that 1000 training feature-path pairs only account for approximately 0.1\% of all possible pairs in the smallest dataset (i.e., DonorsChoose). We also observe that random data selection strategy demonstrates competitive performance. It is expected as all test sets are randomly drawn from three datasets, respectively.

We further employ the saliency method~\cite{simonyan2014saliency} to measure the importance of each input feature to the $IQ$ model, by assessing their respective gradient concerning the predicted $IQ$ value. As shown in Figure~\ref{fig:saliency}, the $IQ$ model benefits the most from Pearson correlation and variance, and the other features also contribute to the $IQ$ to varying degrees with different datasets. This confirms that the chosen features are essential to a high-quality $IQ$ estimator in \method.







\subsubsection{Effectiveness of \method's Feature Importance Estimator.}

Second, we evaluate the accuracy of our $FI$ estimator. The ground truth is established by materializing the test join paths and computing the corresponding $FI$ scores using AutoGluon. We measure and report the MAE between the estimated values and the ground truth values. We also vary the size of the training set from 50 to 1000. 
To make sense of the MAE of our $FI$ estimator, we compare it against three alternative designs: (1) a random sampling method that randomly choose $|f|$ destination features for estimating $FI$. We ensure that at least one sample is drawn from each cluster; (2) a top centroid-based method that allocate budgets uniformly across clusters. Within each cluster, we select the top central features; and (3) a top cover method that equally distributes budgets across clusters. The objective within each cluster is to identify points that minimize the \textit{coverage radius}. 

Overall, as shown in Figure~\ref{fig:fi_exp}, we observe that \method's estimation quality (i.e., lower MAE) improves with an increasing number of data points collected from feature clusters. \method outperforms the other baselines on \textsf{School}, \textsf{Fraud Detection} and \textsf{Air} datasets, and its performance on \textsf{DonorsChoose} and \textsf{Poverty} is also very close to the best performing \textit{top cover} method. When the given budget is small (i.e., $|f|\leq 100)$, the random sampling method is slightly more effective on \textsf{DonorsChoose} and \textsf{School} datasets. In fact, this is consistent with our observation in the above integration quality evaluation. Hence, one could take a hybrid approach that randomly samples the first 50 or 100 feature-path pairs and then adapts to the proposed data selection strategy when establishing the $FI$ estimator and the $IQ$ model.

\subsubsection{Effectiveness and Efficiency of \method's Search Algorithm.}

Third, we evaluate the effectiveness and efficiency of \method's search algorithm. We report the final Accuracy or F1 as well as the compute time spent on search. To make sense of the measurements, we compare \method's search algorithm against two baselines: (1) an exhaustive method that uses BFS to exhaustively search and estimate all possible combinations of destination features and integration paths; and (2) a greedy method that chooses features with highest estimated $FI$ for candidate features. For each candidate feature, it further selects the path with the fewest number of joins (i.e., shortest integration path). If more than one integration paths are the shortest, the greedy method picks the one with the highest estimated $IQ$. In this experiment, we vary the feature budget from 1 to 10, to be consistent with Table~\ref{tab.main_exp}.

\begin{table*}[htbp]
\centering
\resizebox{1\linewidth}{!}{\begin{tabular}{c||c|c|c|c|c|c||c|c|c|c|c|c||c|c|c|c|c|c}
\toprule
\multirow{1}{*}{\textbf{Methods}} & \multicolumn{6}{c||}{\textsf{Exhuasive}} & \multicolumn{6}{c||}{\textsf{Greedy}} & \multicolumn{6}{c}{\textsf{\method}} \\
\hline
\hline
\multirow{2}{*}{\textbf{Datasets}} & \multicolumn{2}{c|}{$|B|$=1} & \multicolumn{2}{c|}{$|B|$=5} & \multicolumn{2}{c||}{$|B|$=10} & \multicolumn{2}{c|}{$|B|$=1} & \multicolumn{2}{c|}{$|B|$=5} & \multicolumn{2}{c||}{$|B|$=10} & \multicolumn{2}{c|}{$|B|$=1} & \multicolumn{2}{c|}{$|B|$=5} & \multicolumn{2}{c}{$|B|$=10} \\
\cline{2-19}
& Score & Time & Score & Time & Score & Time & Score & Time & Score & Time & Score & Time & Score & Time & Score & Time  & Score & Time  \\
\hline
School & 0.823 & 45.33 & 0.880 & 45.53 & 0.885 & 45.69 & 0.720 & 3.02 & 0.801 & 3.08 & 0.805 & 3.16 & 0.823 & 8.046 & 0.891 & 8.194 & 0.880 & 8.392   \\
\hline
DonorsChoose & 0.822 & 71.99 & 0.954 & 72.21 & 0.961 & 72.52 & 0.855 & 1.77 & 0.925 & 1.88 & 0.955 & 2.05 & 0.822 & 32.51 & 0.954 & 32.70 & 0.961 & 32.97    \\
\hline
Fraud Detection & 0.403 & 61.54 & 0.513 & 61.78 & 0.538 & 61.80 & 0.403 & 3.55 & 0.492 & 3.62 & 0.478 & 3.76 & 0.403 & 32.58 & 0.513 & 32.66 & 0.538 & 32.81 \\
\hline
Poverty & 8222.34 & 36.66 & 7739.44 & 36.75 & 7473.87  & 36.85 & 8222.34 & 1.34 & 7885.65 & 1.39 & 7716.36 & 1.43 & 8222.34 & 16.15 & 7322.44 & 16.20 & 7182.38 & 16.26 \\
\hline
Air &  1.044 & 16.79 & 0.7619 & 16.92 & 0.7158 & 17.15 & 1.090 & 2.06 & 0.766 & 2.20 & 0.732 & 2.33 & 1.101 & 7.19 & 0.762 & 7.32 & 0.715 & 7.53  \\
\bottomrule
\end{tabular}}
\caption{Effectiveness and efficiency of \method's search algorithm (Time in seconds).}
\label{tab.search_exp}
\end{table*}

\begin{figure*}
  \centering
  \begin{subfigure}[b]{0.195\textwidth}
    \centering
    \includegraphics[width=\textwidth]{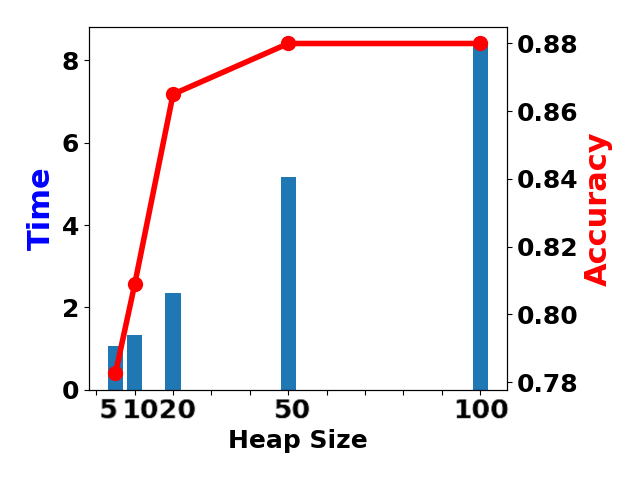}
    \caption{School}
    \label{fig:school}
  \end{subfigure}
  \hfill
  \begin{subfigure}[b]{0.195\textwidth}
    \centering
    \includegraphics[width=\textwidth]{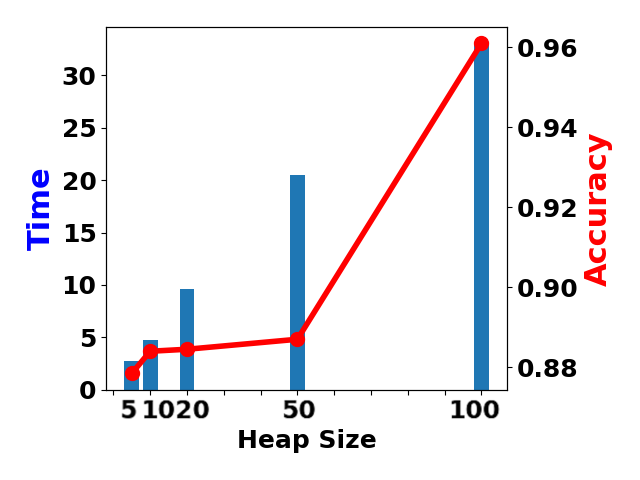}
    \caption{DonorsChoose}
    \label{fig:donor}
  \end{subfigure}
  \hfill
  \begin{subfigure}[b]{0.195\textwidth}
    \centering
    \includegraphics[width=\textwidth]{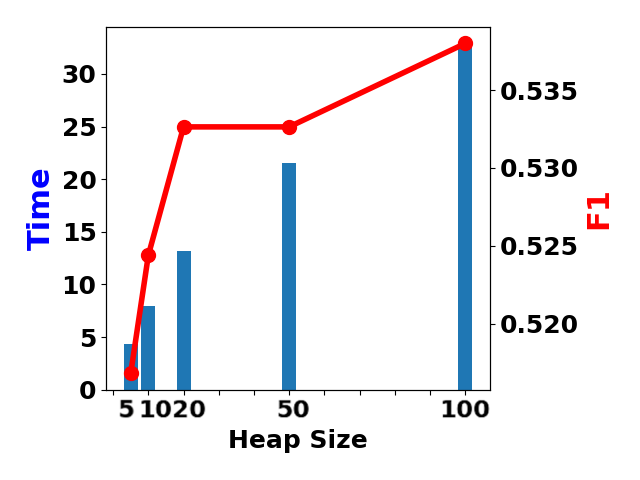}
    \caption{Fraud Detection}
  \end{subfigure}
  \hfill
  \begin{subfigure}[b]{0.195\textwidth}
    \centering
    \includegraphics[width=\textwidth]{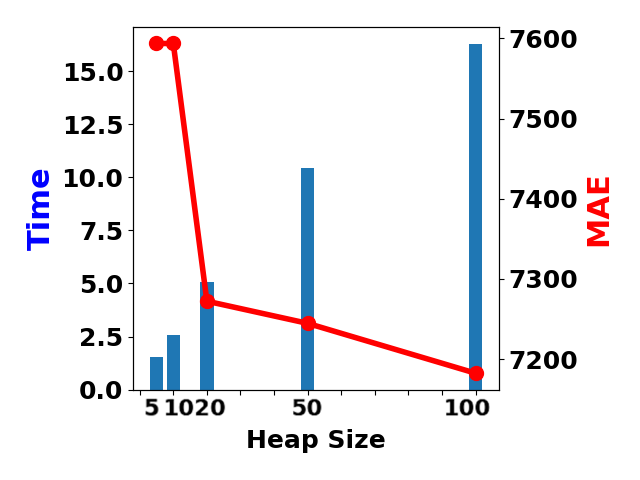}
    \caption{Poverty}
  \end{subfigure}
  \hfill
  \begin{subfigure}[b]{0.195\textwidth}
    \centering
    \includegraphics[width=\textwidth]{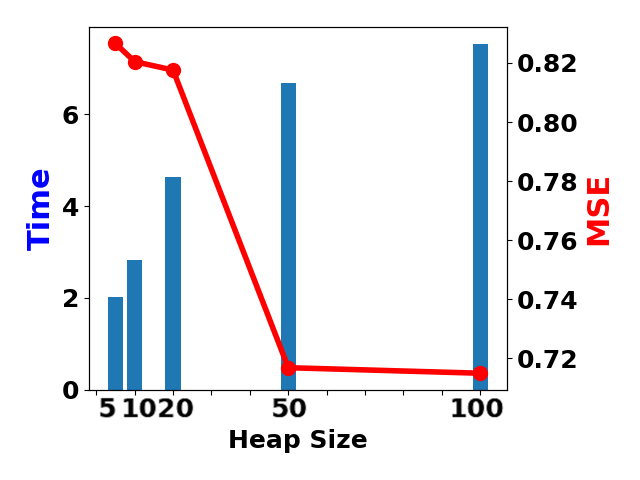}
    \caption{Air}
  \end{subfigure}
  \vspace{-10pt}
  \caption{Varying heap size for \method's search algorithm.}
  \label{fig:heapsize}
\end{figure*}

As shown in Table~\ref{tab.search_exp}, \method's search algorithm significantly outperforms the greedy algorithm and achieves similar performance to the exhaustive algorithm in Accuracy or F1. Regarding the efficiency, \method's search algorithm also achieves 2-5 times speedup compared to the exhaustive algorithm. Moreover, \method's compute time, on average, increases 2\% when increasing the feature budget from 1 to 10. This shows that \method's search algorithm is scalable to a large number of features allowed for augmentation.

We further evaluate the performance of \method's search algorithm by varying heap sizes $H$ from 5 to 100 and report the final model performance and the search compute time. As shown in Figure~\ref{fig:heapsize}, the model performance shows solid improvements with an increasing heap size, while the latency only increases sub-linearly. This showcases that \method's search and refinement algorithms are effective and efficient.

\subsection{End-to-End Latency}

\small
\begin{table*}[htbp]
\centering
\resizebox{1\linewidth}{!}{\begin{tabular}{c||c|c||c|c||c|c||c|c||c|c}
\toprule
\textbf{Datasets} & \multicolumn{2}{c||}{\textsf{School}} & \multicolumn{2}{c||}{\textsf{DonorsChoose}} & \multicolumn{2}{c||}{\textsf{Fraud Detection}} & \multicolumn{2}{c||}{\textsf{Poverty}} & \multicolumn{2}{c}{\textsf{Air}} \\
\hline
\textbf{Methods} & Actual Score & Time & Actual Score & Time & Actual Score & Time & Actual Score & Time & Actual Score & Time \\
\hline
\hline
DFS         & 0.692 & 53  & 0.854 & 140  & 0.084 & 95  & 13077.66 & 184  & 0.9425 & 121
\\
\hline
ARDA        & 0.794 & 1,517  & 0.901 & 1,707  & 0.450 & 921  & 12164.23 & 3,155  & 0.7614 & 3,003 \\
\hline
AutoFeature & 0.723 & 3,488  & 0.896 & 3,403  & 0.162 & 3,942  & 13077.66 & 9,268  & 1.2016 & 4,747 \\
\hline
METAM       & 0.816 & 1,723  & 0.820 & 1,120  & 0.464 & 1,912  & 12786.82 & 1,651  & 0.7618 & 2,249 \\
\hline
\method     & 0.880 & 1,406  & 0.961 & 1,602  & 0.540 & 1,872  & 7182.38 & 1,059  & 0.7154 & 1,488  \\
\bottomrule
\end{tabular}}
\caption{End-to-end latency (Time in seconds).}
\label{tab:e2e_exp}
\end{table*}

In Table~\ref{tab:e2e_exp}, we report the end-to-end latency of \method and other baselines over three datasets. The feature budget is set to 10 to be consistent with Table~\ref{tab.main_exp}. Specifically, \method is 2.1\% and 18.4\% faster than METAM on \textsf{Fraud Detection} and \textsf{School} datasets, and 6.2\% faster than ARDA on \textsf{DonorsChoose} dataset. This shows that \method not only delivers the best model performance but also has lower latency compared to the best performing baselines.


\section{Related Work}
\label{sec:related}

\para{Feature Augmentation.} 
Automatic feature augmentation has been intensively studied and shown effective for various ML tasks including classification, regression, clustering, etc. Kumar et al.~\cite{10.1145/2882903.2882952} focus on the problem of when to avoid primary key-foreign key join without sacrificing the performance of the model. Deep Feature Synthesis (DFS)~\cite{dfs} is one of the early work in generating features for relational datasets. DFS follows joinable relationships in the data to a base table, and then sequentially applies mathematical functions along that path to create the final features. ARDA~\cite{arda} is a feature augmentation framework that leverages existing data discovery tools~\cite{aurum} to score the candidate tables and uses a heuristic algorithm to select features. Hence, the performance of ARDA heavily relies on the scores given by Aurum, which are not always accurate as they are not model-aware. 

Recently, AutoFeature~\cite{autof} leverages reinforcement learning (RL) for feature augmentation. It uses either a multi-armed bandits or a branch Deep Q Networks to choose between exploring features that could lead to performance improvement and exploiting features that are rarely selected. METAM~\cite{metam} is a goal-oriented framework that also utilizes the multi-arm bandit method for feature discovery and augmentation to achieve the desired performance of a given task. It uses data properties, utility functions from downstream tasks and the given candidate dataset to drive the discovery and augmentation process. Coreset selection~\cite{10.14778/3561261.3561267,DBLP:journals/pacmmod/ChaiL0FM0L023} aims to select a high-quality coreset without materializing the augmented table and uses weighted gradients of the selected coreset to approximate the full gradient of the entire train dataset.


\para{Feature Selection.}
Feature selection~\cite{featureselectionsurvey2014,featureselectiondata2017} aims to identify a subset of the most relevant features from a large dataset for a specific ML task, thereby reducing feature dimensionality. In contrast, feature augmentation involves increasing the number of features to enhance the model's predictive capability. The ranking-based methods evaluate the importance of each feature using statistical metrics like correlation coefficient or mutual information with the labels~\cite{guyon2003introduction} and then select the top relevant features. However, ranking-based methods may yield a sub-optimal subset with highly correlated variables, where a smaller subset would suffice. 

Sequential selection and heuristic search algorithms optimize a utility gain objective function for feature subset evaluation. Due to the NP-hard nature of evaluating $2^N$ subsets, these methods find sub-optimal subsets heuristically. Sequential algorithms~\citep{floating1994} either add or remove features from an initial empty (or full) set until the utility gain is maximized, using a criterion to incrementally enhance the utility gain  with fewer features. Heuristic algorithms~\citep{goldberg2013genetic}, on the other hand, explore different subsets, either by searching within a space or generating solutions to optimize the utility gain. More recently, Cherepanova et al.~\cite{cherepanova2023performance} propose an input-gradient-based analogue of Lasso for neural networks that outperforms classical feature selection methods on challenging problems such as selecting from corrupted or second-order features in the tabular deep learning setting. Liu et al.~\cite{featureselectionreinforce2023} reformulate feature selection with a RL framework by regarding each feature as an agent.

\para{Table Discovery.}
Table discovery systems~\cite{aurum,zhu2019josie,DBLP:conf/sigmod/ZhangI20,DBLP:journals/pvldb/Dong0NEO23,khatiwada2023santos,hu2023automatic,DBLP:conf/sigmod/Fan00M23} find joinable and unionable tables from data repositories to boost downstream applications, such as related table discovery, table QA, and ML models. Modern techniques~\cite{DBLP:conf/sigmod/ZhangI20,DBLP:conf/icde/GongZGF23} allow users to interactively search useful joinable datasets by specifying requirements in the form of data profiles or examples. However, without the information of ML models, it is difficult to identify which tables and join paths are useful to the downstream ML models.
\section{Conclusion}
\label{sec:conclusion}

In this paper, we propose a framework \method for automatic feature augmentation. \method tackles the intractable problem of feature augmentation by decomposing it into feature importance and integration quality. We design a lightweight clustering-based method for feature importance estimation and an LSTM-based integration quality model. We further introduce an efficient feature path search algorithm that exploits both estimated feature importance and integration quality to identify high-quality features and their optimized join paths for integration under a given budget. Our experiments on three public datasets demonstrate the effectiveness of \method.

\balance
\bibliographystyle{ACM-Reference-Format}
\bibliography{afp}

\end{document}